\documentclass[american]{article}
\usepackage[T1]{fontenc}
\usepackage[latin9]{inputenc}
\usepackage[active]{srcltx}
\usepackage{babel}
\usepackage{cprotect}
\usepackage{booktabs}
\usepackage{units}
\usepackage{url}
\usepackage{amsmath}
\usepackage{amsthm}
\usepackage{amssymb}
\usepackage{graphicx}
\PassOptionsToPackage{normalem}{ulem}
\usepackage{ulem}
\usepackage{multirow}
\usepackage[bookmarks=true,bookmarksnumbered=false,bookmarksopen=false,
]{hyperref}
\usepackage[switch]{lineno}
\usepackage[hypcap]{caption}  
\makeatletter

\numberwithin{equation}{section}
\numberwithin{figure}{section}
\theoremstyle{plain}
\newtheorem{thm}{\protect\theoremname}
\theoremstyle{remark}
\newtheorem{rem}[thm]{\protect\remarkname}
\theoremstyle{plain}
\newtheorem{prop}[thm]{\protect\propositionname}

\usepackage{fullpage}
\usepackage{authblk}

\usepackage[section]{placeins}
\usepackage{hyphenat}
\usepackage{fancyvrb}
\usepackage{subcaption}
\usepackage{bbm} 
\usepackage{relsize}
\usepackage{amsbsy}
\usepackage{textcomp}
\usepackage{pgf} 
\usepackage{xcolor}
\usepackage{array}
\usepackage{tabularx}
\usepackage{float}
\usepackage{adjustbox}
\usepackage{enumitem}
\usepackage{natbib}
\usepackage{algorithm}
\usepackage{algorithmic}

\makeatother

\providecommand{\propositionname}{Proposition}
\providecommand{\remarkname}{Remark}
\providecommand{\theoremname}{Theorem}

\title{Market-Informed Valuation of GMMB Riders with Surrender Options under a Heston Stochastic-Local Volatility Model}
\author[1]{\textsc{Ludovic Gouden\`ege}\thanks{Email: \texttt{ludovic.goudenege@math.cnrs.fr}. ORCID: \url{https://orcid.org/0000-0002-6449-5888}.}}
\author[2]{\textsc{Andrea Molent}\thanks{Corresponding author. Email: \texttt{andrea.molent@uniud.it}. ORCID: \url{https://orcid.org/0000-0002-0887-826X}.}}
\author[3]{\textsc{Xiao Wei}\thanks{Email: \texttt{weixiao@cufe.edu.cn}. ORCID: \url{https://orcid.org/0009-0008-0835-7238}.}}
\author[2]{\textsc{Antonino Zanette}\thanks{Email: \texttt{antonino.zanette@uniud.it}. ORCID: \url{https://orcid.org/0000-0001-5958-9014}.}}
\affil[1]{Laboratoire de Math\'ematiques et Mod\'elisation d'\'Evry (LaMME), CNRS UMR 8071, Universit\'e Paris-Saclay \`Evry, France}
\affil[2]{Dipartimento di Scienze Economiche e Statistiche, Universit\`a degli Studi di Udine, Udine, Italy}
\affil[3]{China Institute for Actuarial Science \& School of Insurance, Central University of Finance and Economics, China}
\date{}

\hypersetup{
	pdftitle={Market-Informed Valuation of GMMB Riders with Surrender Options under a Heston Stochastic-Local Volatility Model},
	pdfauthor={Ludovic Goudenege; Andrea Molent; Xiao Wei; Antonino Zanette},
	pdfsubject={Valuation of surrenderable GMMB riders under local and stochastic-local volatility},
	pdfkeywords={variable annuities, guaranteed minimum maturity benefits,
		surrender options, Heston stochastic-local volatility,
		hybrid tree/finite-difference method}
}

\newcommand{\dd}{\,\mathrm{d}}
\newcommand{\E}{\mathbb{E}}
\newcommand{\Q}{\mathbb{Q}}

\newcommand{\calA}{\mathcal{A}}

\newcommand{\sigLV}{\sigma_{\mathrm{LV}}}
\newcommand{\sigDup}{\sigma_{\mathrm{Dup}}}

\newcommand{\barrho}{\bar\rho}

\DeclareMathOperator*{\esssup}{ess\,sup}

\newcommand{\PriceDecimals}{2}
\newcommand{\pnum}[1]{\pgfmathprintnumber[fixed zerofill,precision=\PriceDecimals]{#1}}
\providecommand{\num}[2][]{\pgfmathprintnumber[#1]{#2}}
\newcommand{\ptime}[1]{%
	\pgfmathprintnumber[
	fixed,
	fixed relative,
	precision=2
	]{#1}%
}
\newcolumntype{F}{r}
\newcolumntype{P}{r}

\begin{document}
	\maketitle
	\begin{flushleft}
		\rule{1\columnwidth}{1pt}
		\par\end{flushleft}
	
	\begin{flushleft}
		{\large\textbf{Abstract}}{\large\par}
		\par\end{flushleft}
	We develop a market-informed valuation framework for
        guaranteed minimum maturity benefit (GMMB) riders with
        rational surrender under the Heston stochastic-local
        volatility (SLV) model. The guarantee is written on the
        fee-deducted account value and is considered both in its
        terminal-only form and in the presence of early surrender
        rights. The Heston SLV specification combines stochastic volatility with a leverage function calibrated to a prescribed local-volatility surface. The leverage surface is obtained through a forward Markovian-projection equation so that, at the model level, the SLV dynamics are constrained to the same one-dimensional marginals as the corresponding local-volatility (LV) model. The latter is used only as a one-factor benchmark, allowing us to isolate the effect of stochastic volatility on continuation values and surrender decisions while preserving the same option-calibrated local-volatility target.
We derive the associated backward pricing equations and propose a hybrid tree/finite-difference algorithm for the SLV model with a calibrated leverage function. Synthetic experiments and a market-informed case study show that SLV and LV valuations are numerically close for terminal-only guarantees, as expected from the common marginal target, whereas materially larger differences can arise once surrender is allowed. These differences are reflected in guarantee values, fair insurance fees and volatility-dependent surrender regions. The results indicate that matching one-date marginals implied by vanilla-option prices does not eliminate model risk for insurance liabilities whose value depends on conditional continuation dynamics and endogenous surrender decisions.
	
	
	
	
	\begin{flushleft}
		\rule{1\columnwidth}{1pt}
		\par\end{flushleft}
	
	\section{Introduction}
	\label{sec:introduction}
	
	Variable annuities combine investment exposure with insurance protection. During the accumulation period, the policyholder's account is linked to an underlying fund, while the insurer provides contractual guarantees against adverse market outcomes. Guaranteed minimum benefits have therefore been studied extensively in actuarial finance, both because they are option-like liabilities and because their value is sensitive to policyholder behaviour, fees and market volatility. General valuation frameworks for variable-annuity guarantees are given, among others, by \citet{Bauer2008}. Related analyses of lapse and surrender features, withdrawal guarantees and participating contracts include \citet{MilevskySalisbury2001}, \citet{MilevskySalisbury2006}, \citet{Bacinello2003}, \citet{Bernard2014}, \citet{HyndmanWenger2014} and \citet{Siu2005}.
	
	This paper focuses on guaranteed minimum maturity benefits (GMMBs) with surrender options. A GMMB guarantees a minimum account value at the end of the accumulation period. In the specification considered here, the embedded put-like guarantee rider may be exercised before maturity, so the policyholder compares its immediate surrender benefit, net of the time-dependent surrender adjustment, with its continuation value. In a constant-volatility setting, this rider can be represented as an American-style put option written on the fee-deducted account value; related analyses provide free-boundary representations, integral pricing formulae and fair insurance fees \citep{Shen2016}. Earlier and related work on equity-linked and variable-annuity contracts with surrender options includes \citet{ShenXu2005}, \citet{Costabile2008} and \citet{Bernard2014}. We adopt this standard guarantee-rider specification, but replace the constant-volatility fund dynamics by volatility models designed to reproduce a non-flat volatility structure. The surrender considered below is therefore the exercise of the guarantee rider, not an empirical model of lapse of the entire variable-annuity account.
	
	Stochastic volatility has already been introduced into closely related surrenderable variable-annuity problems. In particular, \citet{KangZiveyi2018} value GMMBs with optimal surrender under Heston stochastic volatility and stochastic interest rates, while \citet{HuhJeonPark2023} study a surrender option under multiscale stochastic volatility; related state-dependent fee incentives are analysed by \citet{MacKayVachonCui2023}. These contributions establish the relevance of stochastic volatility for surrenderable guarantees. The question addressed here is narrower: the local-volatility model and the Heston stochastic-local volatility model are constrained to the same one-dimensional marginal information generated by a common local-volatility target, so that pricing differences beyond calibration and discretisation errors can be interpreted in terms of the conditional dynamics that enter the stopping problem.
	
	The motivation for doing so is standard in option pricing but particularly relevant for long-dated insurance liabilities. A constant volatility parameter cannot reproduce the observed dependence of implied volatility on strike and maturity. Local-volatility (LV) models were introduced precisely to address this limitation by replacing the constant volatility with a deterministic function of time and state. The constructions of \citet{DermanKani1994} and \citet{Dupire1994} provide the classical link between European option prices and the local volatility surface. Numerical implementation and regularisation of the Dupire inversion have been studied in several forms, including finite-difference and spline-based approaches; see, for example, \citet{AchdouPironneau2005}, \citet{Crepey2003}, \citet{Itkin2020} and recent learning-based approaches such as \citet{Wang2025} and \citet{MolentVellekoop2026}. Within the actuarial literature, \citet{DeelstraRayee2013} investigate the valuation of variable-annuity guarantees in a local-volatility framework with stochastic interest rates. In the present paper, the LV specification serves as the one-factor benchmark associated with the prescribed or market-calibrated target local-volatility surface.
	
	In the market-data application, the observable inputs are constructed
	consistently from derivative prices on the reference index: the target LV
	surface is calibrated from European put prices, while matched call--put quotes
	are used to infer the deterministic risk-free and dividend-yield term
	structures through put--call parity. 
	
	Stochastic-local volatility (SLV) models add a second layer of modelling flexibility. In this paper we use a Heston SLV specification, which combines the Heston state variable $V_t$ introduced by \citet{Heston1993} with a leverage function that preserves consistency with a target local-volatility surface. The theoretical basis is Markovian projection: different multidimensional diffusions may have the same one-dimensional marginal distributions, as in \citet{Gyongy1986}, but different conditional dynamics. This is the reason why a local-volatility model and an SLV model can be calibrated to the same vanilla option information and yet produce different prices for path-dependent or early-exercise claims. Practical accounts of stochastic and stochastic-local volatility modelling and calibration include \citet{Gatheral2006}, \citet{RenMadanQian2007}, \citet{GuyonHenryLabordere2013} and \citet{Bergomi2016}.  	In the market application, we refer to this Heston SLV specification as \emph{market-informed}: the Heston parameters are informed by volatility-market and joint equity--volatility information, while the leverage function is calibrated to the common LV target.
	
	The calibration of the leverage function is itself a nonlinear problem, since the conditional moment entering the projection identity is generated by the joint law of the SLV process. Approaches based on partial differential equations (PDEs) have therefore received considerable attention. \citet{WynsInTHout2018} introduce an adjoint semidiscretisation of the forward Kolmogorov equation that yields exact calibration of the semidiscretised SLV model to the corresponding semidiscretised local-volatility model for non-path-dependent European payoffs. \citet{SaporitoYangZubelli2019} cast SLV calibration as an inverse problem and use regularisation techniques to obtain stable leverage surfaces, especially in low-density regions where the conditional expectation is numerically fragile. The mathematical well-posedness of calibrated stochastic-local volatility dynamics is also delicate, since the leverage function depends on the conditional law of the state process; see, for example, \citet{JourdainZhou2020}.
	
	The distinction between LV and SLV is especially relevant for surrenderable guarantees. European values are determined by one-date marginal distributions, whereas early surrender depends on conditional continuation values and, under SLV, on the current volatility state. Hence a local-volatility projection may reproduce the vanilla surface while producing a different surrender policy. Rather than benchmarking SLV against an arbitrary constant-volatility model, we therefore ask how large this residual effect is when LV and SLV are constrained to a common non-flat volatility structure. We measure it through guarantee prices, fair insurance fees and surrender boundaries.
	
	The numerical methodology is related to the hybrid tree/finite-difference literature for stochastic-volatility and stochastic-interest-rate models. A hybrid construction for the Heston model is developed in \citet{BrianiCaramellinoZanette2017Heston}, while Heston--Hull--White-type extensions are considered in \citet{BrianiCaramellinoZanette2017}. For the trinomial construction itself, we follow the recombining moment-matching tree used in the variable-annuity literature, in particular in \citet{Molent2020} and in the trinomial-tree construction of \citet{GoudenegeMolentZanette2021}.
	
	The contribution is threefold. First, we propose a forward--backward
	hybrid tree/finite-difference implementation tailored to the
	leverage-calibrated SLV problem: a recombining moment-matching trinomial
	tree represents the state variable $V_t$, a conservative forward finite-volume
	step propagates the joint law used to calibrate the leverage function,
	and one-dimensional backward finite-difference solves implement
	continuation and optimal surrender after a decorrelating transformation
	in the correlated case. Second, we formulate a matched-marginal
	comparison between an LV model and a Heston SLV model for a surrenderable
	GMMB rider. For European payoffs, residual LV--SLV differences provide
	a numerical check of the marginal calibration, whereas persistent
	differences for surrenderable claims reflect the effect of conditional
	transition dynamics that is not fixed by the common local-volatility
	target. Third, we quantify these differences through guarantee prices,
	fair insurance fees and surrender regions, first with the fixed synthetic
	local-volatility surface of \citet{Wang2025} and then with a surface
	calibrated to EURO STOXX 50 option prices, using VSTOXX and joint
	equity--volatility information to inform the structural Heston parameters.
	
	The rest of the paper is organised as follows. Section \ref{sec:market} presents the LV and SLV financial market models, together with the pricing and calibration PDEs. Section \ref{sec:contract} describes the GMMB contract and its early-surrender formulation. Section \ref{sec:hybrid} describes the hybrid tree/finite-difference method. Section \ref{sec:numerics} reports the numerical experiments, including the synthetic convergence analysis, pricing tables, fair fees and surrender regions, followed by the market-data analysis. Section \ref{sec:conclusion} concludes. Numerical  details are collected in Appendix~\ref{app:numerical-details}.
	
	\section{The Heston SLV market model}
	\label{sec:market}
	
	Fix a finite horizon $T>0$. Let $(\Omega,\mathcal{F},(\mathcal{F}_t)_{0\leq t\leq T},\Q)$ be a filtered probability space satisfying the usual conditions. All prices are computed under the risk-neutral measure $\Q$. The risk-free rate $r_t$ and the dividend yield $q_t$ are deterministic functions of time. The money-market account is
	\[
	B_t=\exp\!\left(\int_0^t r_u\,\dd u\right),
	\]
	and, for $0\leq t\leq u\leq T$, we write
	\[
	D_r(t,u)=\exp\!\left(-\int_t^u r(z)\,\dd z\right),
	\qquad
	D_q(t,u)=\exp\!\left(-\int_t^u q(z)\,\dd z\right).
	\]
	Let $(S_t)_{0\leq t\leq T}$ denote the reference price-index or fund
	process before the deduction of the insurance fee, with initial level
	$S_0>0$.

	Let $(W_t)_{t\geq0}$ be a one-dimensional Brownian motion and let $\sigLV(t,s)>0$ be a deterministic local-volatility function. Under the local-volatility model the reference fund dynamics are
	\begin{equation}
		\frac{\dd S_t}{S_t}=\bigl(r_t-q_t\bigr)\dd t+\sigLV(t,S_t)\dd W_t.
		\label{eq:lv-sde}
	\end{equation}
	Throughout this section, lowercase $u$ denotes a generic contingent-claim value. The market model is identified by the corresponding generator and state variables. For a sufficiently regular claim value $u(t,s)$, define the LV generator by
	\begin{equation}
		\calA^{\mathrm{LV}}_t u
		=\bigl(r_t-q_t\bigr)s\partial_s u
		+\frac12 \sigLV^2(t,s)s^2\partial_{ss}u.
		\label{eq:lv-generator}
	\end{equation}
	The associated backward pricing PDE in the continuation region is then
	\begin{equation}
		\partial_t u(t,s)+\calA^{\mathrm{LV}}_t u(t,s)-r_t u(t,s)=0.
		\label{eq:lv-backward}
	\end{equation}
	The LV pricing equation is a backward terminal-value problem. For a terminal payoff $\Psi(S_T)$, the corresponding backward terminal condition is
	\begin{equation}
		u(T,s)=\Psi(s).
		\label{eq:lv-generic-terminal}
	\end{equation}
	If early surrender is allowed, this terminal condition is complemented by a pointwise obstacle condition at the admissible surrender dates. Thus the datum used by the LV pricing problem is an ordinary terminal payoff. In the numerical work below the LV model is used only through backward pricing solves; no forward density propagation is performed for the LV benchmark.
	
	The calibration of $\sigLV$ from European option prices is standard. If $C(K,T)$ denotes the time-zero price of a call option with strike $K$ and maturity $T$, then, under the usual smoothness and no-arbitrage assumptions, we denote the Dupire local volatility by $\sigDup$ and write
	\begin{equation}
		\sigDup^2(K,T)
		=\frac{2\left(\partial_T C(K,T)+q_T C(K,T)+\bigl(r_T-q_T\bigr)K\partial_K C(K,T)\right)}{K^2\partial_{KK}C(K,T)}.
		\label{eq:dupire}
	\end{equation}
	The local-volatility surface used in the numerical experiments is either prescribed synthetically or, in the market application, obtained beforehand from European option prices through a calibrated recombining market model that is subsequently represented as a dense local-volatility grid. Classical and learning-based approaches to local-volatility construction include \citet{Dupire1994}, \citet{DermanKani1994}, \citet{AchdouPironneau2005}, \citet{Crepey2003}, \citet{Wang2025} and \citet{MolentVellekoop2026}.
	
	The Heston SLV model augments the local-volatility dynamics with the Heston state variable $(V_t)_{t\geq0}$. Let $V_0>0$, $\kappa_V>0$, $\theta_V>0$ and $\omega>0$ denote, respectively, its initial level, mean-reversion speed, long-run level and diffusion coefficient. Let $\rho\in[-1,1]$ denote the instantaneous correlation between two Brownian motions $W^S$ and $W^V$, and let $L(t,s)>0$ denote the deterministic leverage function. The Heston SLV specification is
	\begin{alignat}{2}
		&\frac{\dd S_t}{S_t} &&=\bigl(r_t-q_t\bigr)\dd t+L(t,S_t)\sqrt{V_t}\dd W^S_t,
		\label{eq:slv-S}\\
		&\dd V_t &&=\kappa_V(\theta_V-V_t)\dd t+\omega\sqrt{V_t}\dd W^V_t,
		\label{eq:slv-V}\\
		&\dd\langle W^S,W^V\rangle_t &&=\rho\,\dd t.
		\label{eq:slv-corr}
	\end{alignat}
	In the pure Heston model, obtained when $L\equiv1$, $V_t$ represents the instantaneous variance. In the SLV specification, instead, the instantaneous volatility of the reference fund is $L(t,S_t)\sqrt{V_t}$, and the corresponding instantaneous variance rate is $L^2(t,S_t)V_t$.
	
	For a sufficiently smooth claim value $u(t,s,v)$, the SLV generator is
	\begin{equation}
		\begin{aligned}
			\calA^{\mathrm{SLV}}_t u
			=&\;\bigl(r_t-q_t\bigr)s\partial_s u
			+\kappa_V(\theta_V-v)\partial_v u \\
			&+\frac12L^2(t,s)v s^2\partial_{ss}u
			+\rho\omega L(t,s)v s\partial_{sv}u
			+\frac12\omega^2v\partial_{vv}u.
		\end{aligned}
		\label{eq:slv-generator}
	\end{equation}
	The backward pricing equation in the continuation region is
	\begin{equation}
		\partial_t u(t,s,v)+\calA^{\mathrm{SLV}}_t u(t,s,v)-r_t u(t,s,v)=0.
		\label{eq:slv-backward}
	\end{equation}
	For a payoff $\Psi$ depending on the terminal fund level, the corresponding backward terminal condition is
	\begin{equation}
		u(T,s,v)=\Psi(s).
		\label{eq:slv-generic-terminal}
	\end{equation}
	This is the case for the GMMB payoff specified in Section~\ref{sec:contract}, since the contractual cash flow is written on the fee-deducted account value. Nevertheless, the SLV continuation value generally depends on $v$ through the dynamics of $V_t$. As in the LV case, early surrender is treated by imposing the obstacle condition at the admissible surrender dates.
	
	Let $p(t,s,v)$ denote the joint probability density of
	$(S_t,V_t)$ under $\Q$. Its forward Kolmogorov equation is
	\begin{equation}
		\begin{aligned}
			\partial_t p
			=&-\partial_s\{\bigl(r_t-q_t\bigr)s p\}
			-\partial_v\{\kappa_V(\theta_V-v)p\} \\
			&+\frac12\partial_{ss}\{L^2(t,s)v s^2p\} \\
			&+\partial_{sv}\{\rho\omega L(t,s)v s p\}
			+\frac12\partial_{vv}\{\omega^2 v p\}.
		\end{aligned}
		\label{eq:slv-forward}
	\end{equation}
	The SLV forward calibration problem is initialized at the deterministic state $(S_0,V_0)$. Writing $\delta$ for the Dirac distribution, the initial condition is
	\begin{equation}
		p(0,s,v)=\delta(s-S_0)\,\delta(v-V_0).
		\label{eq:slv-forward-initial}
	\end{equation}
	The initial condition \eqref{eq:slv-forward-initial} is used in the SLV forward calibration pass to propagate the joint law and compute the conditional moment entering the leverage update.
	The leverage function is chosen so that the SLV model reproduces the target local variance rate by Markovian projection. If $\sigLV$ denotes the target local volatility, the calibration condition is
	\begin{equation}
		L^2(t,s)\,\E^{\Q}\left[V_t\mid S_t=s\right]=\sigLV^2(t,s),
		\qquad
		L(t,s)=\frac{\sigLV(t,s)}{\sqrt{\E^{\Q}[V_t\mid S_t=s]}}.
		\label{eq:leverage-condition}
	\end{equation}
	This identity is the link between the forward calibration step and the backward pricing step. In practice, the conditional expectation in \eqref{eq:leverage-condition} is evaluated on the numerical grid from the joint distribution propagated by the hybrid method.
	
	If the LV and SLV specifications use the same deterministic discounting and
	induce the same distribution of $S_t$ at every fixed date, they assign the same
	value to European payoffs depending only on $S_t$. This marginal consistency
	does not extend, in general, to surrenderable claims. Their value is determined
	by a Snell envelope and therefore depends on conditional transition dynamics
	across surrender dates, which are not fixed by the collection of one-date
	marginals. This distinction provides the conceptual basis for the numerical
	comparison below: terminal-only payoffs test the marginal projection, whereas
	the surrender option probes conditional continuation dynamics.

	\section{The GMMB contract}
	\label{sec:contract}
	
	We consider the put-like guarantee component of a variable annuity contract.
	The rider has maturity $T$, guaranteed amount $G>0$, initial account value
	$F_0>0$, continuous insurance fee $c\geq0$ and surrender-adjustment parameter
	$\kappa_s\geq0$. The corresponding policy account value is represented as
	\begin{equation}
		F_t^{(c)}=F_0\frac{S_t}{S_0}e^{-ct}.
		\label{eq:account-value}
	\end{equation}
	Other account conventions are possible. We retain the account convention of
	\citet{Shen2016}, $F_t^{(c)}\propto e^{-ct}S_t$, but extend their setting to
	deterministic dividend yields and time-dependent interest rates. In
	\citet{Shen2016} the reference asset is non-dividend-paying and the
	risk-free rate is constant; here $q_t$ enters the dynamics of the reference
	asset and dividends are not separately reinvested in the policy account.
	
	Throughout this section, uppercase $U$ is reserved for GMMB guarantee values. The superscripts $\mathrm{LV}$ and $\mathrm{SLV}$ identify the market model, while the subscript $\mathrm{Sur}$ is used when the surrender option is present. Under the valuation decomposition adopted here, surrender terminates the guarantee rider and is not modelled as surrender of the entire investment account.
	
	We first consider the GMMB guarantee without surrender, for which the guarantee payoff is received only at maturity $T$. Its terminal payoff is
	\begin{equation}
		\Phi_T(S_T;F_0,c)=\left(G-F_T^{(c)}\right)^+
		=\left(G-F_0\frac{S_T}{S_0}e^{-cT}\right)^+.
		\label{eq:terminal-guarantee-payoff}
	\end{equation}
	Thus, the generic terminal payoff $\Psi$ introduced in Section~\ref{sec:market}
	is here specified as
	\[
	\Psi(s)=\Phi_T(s;F_0,c).
	\]
	In the SLV model, the state also contains $V_t$ and the corresponding time-$t$ value is
	\begin{equation}
		U^{\mathrm{SLV}}(t,s,v)
		=\E^{\Q}\left[D_r(t,T)\Phi_T(S_T;F_0,c)\mid S_t=s,\,V_t=v\right].
		\label{eq:european-gmmb-slv-mean}
	\end{equation}
	The corresponding LV value is obtained by conditioning only on the current fund level and using the LV dynamics of Section~\ref{sec:market}. When the two models are calibrated to the same projected local-volatility structure, these terminal-only values should be close up to calibration and numerical errors.
	
	We next consider the same guarantee with a surrender option. Let $\mathcal D\subset(0,T]$ denote the admissible surrender dates, with $T\in\mathcal D$. The contractual payoff at a surrender time is
	\begin{equation}
		\Phi(t,S_t)=
		\begin{cases}
			\displaystyle
			\left(G-e^{-\kappa_s(T-t)}F_t^{(c)}\right)^+, & t<T,\\[1ex]
			\Phi_T(S_T;F_0,c), & t=T.
		\end{cases}
		\label{eq:surrender-contractual-payoff}
	\end{equation}
	The exponential factor $e^{-\kappa_s(T-t)}$ represents the time-dependent surrender adjustment and converges to one at maturity. Throughout the paper, $\kappa_s$ is therefore referred to as the surrender-adjustment parameter. It is a continuously compounded adjustment rate that modifies the account value entering the guaranteed shortfall; it is not a conventional penalty applied to the cash value of the policy account. At time $t$, the corresponding fractional reduction of the account entering the rider payoff is $1-e^{-\kappa_s(T-t)}$. Under this convention, increasing $\kappa_s$ increases the immediate value of the guarantee rider. Accordingly, $\kappa_s$ should not be interpreted as a conventional surrender penalty applied to the policy account. Exercise terminates only the guarantee rider, while the policy account remains in force until maturity.
	
	For $0\leq t\leq T$, let $\mathcal{T}^{\mathcal D}_{t,T}$ be the set of $\Q$-stopping times taking values in $\mathcal D\cap[t,T]$. In the SLV model, the guarantee value with surrender is the Snell envelope
	\begin{equation}
		U^{\mathrm{SLV}}_{\mathrm{Sur}}(t,s,v)
		=\esssup_{\tau\in\mathcal{T}^{\mathcal D}_{t,T}}
		\E^{\Q}\left[D_r(t,\tau)\Phi(\tau,S_\tau)\mid S_t=s,\,V_t=v\right].
		\label{eq:snell-envelope}
	\end{equation}
	Continuous surrender corresponds to $\mathcal D=(0,T]$. For the numerical approximation, let $N$ be the number of admissible surrender dates and use the finite set $\mathcal D_N=\{\tau_1,\ldots,\tau_N=T\}$ with $0<\tau_1<\cdots<\tau_N=T$; $t=0$ remains a continuation date. In the reported experiments these surrender dates coincide with the positive nodes of the computational time grid introduced in Section~\ref{sec:hybrid}. The optimisation is a risk-neutral rational-surrender convention and should not be interpreted as an empirical lapse model.
	
	For continuous surrender on $(0,T]$, the value $U^{\mathrm{SLV}}_{\mathrm{Sur}}(t,s,v)$ satisfies
	\begin{equation}
		\max\left\{
		\partial_t U^{\mathrm{SLV}}_{\mathrm{Sur}}
		+\calA^{\mathrm{SLV}}_t U^{\mathrm{SLV}}_{\mathrm{Sur}}
		-r_tU^{\mathrm{SLV}}_{\mathrm{Sur}},\;
		\left(G-e^{-\kappa_s(T-t)}F_0\frac{s}{S_0}e^{-ct}\right)^+
		-U^{\mathrm{SLV}}_{\mathrm{Sur}}
		\right\}=0,
		\label{eq:slv-obstacle}
	\end{equation}
	with terminal condition
	\begin{equation}
		U^{\mathrm{SLV}}_{\mathrm{Sur}}(T,s,v)=\Phi_T(s;F_0,c).
		\label{eq:slv-terminal}
	\end{equation}
	The immediate surrender payoff does not depend directly on $v$, but the continuation value does. Therefore, in the SLV model the surrender boundary generally depends on the current value of $V_t$. For a discrete set of admissible surrender dates, the pricing PDE holds between successive dates and the obstacle projection is applied only on $\mathcal D$.
	
For a given initial account value $F_0$, we follow the rider decomposition in
\citet{Shen2016}: surrender terminates the guarantee component, while the
fee-deducted account component is valued to maturity. Let $X_0$ denote the
initial state of the market model, namely $X_0=S_0$ under LV and
$X_0=(S_0,V_0)$ under SLV. For $M\in\{\mathrm{LV},\mathrm{SLV}\}$, the fair
insurance fee $c_M^*$ is defined by
\begin{equation}
	F_0 = \E^{\Q}\left[D_r(0,T)F_T^{(c_M^*)}\right]
	+ U^M_{\mathrm{Sur}}\left(0,X_0;F_0,c_M^*,\kappa_s\right),
	\label{eq:fair-fee-definition}
\end{equation}
where $U^M_{\mathrm{Sur}}$ denotes the surrenderable guarantee value under
model $M$. Since $F_t^{(c)}$ is proportional to the traded price exposure and
the fee is deterministic, the first term in
\eqref{eq:fair-fee-definition} is $F_0e^{-c_M^*T}D_q(0,T)$. Under constant $q$
this reduces to $F_0e^{-(q+c_M^*)T}$. Thus the account component and its
deterministic fee deduction remain in the maturity value even when the
guarantee rider has been surrendered.

We first establish a model-independent monotonicity property that will be used
to interpret the numerical results in both experiments. The argument depends
only on the contractual payoff, the positivity of the fee-deducted account and
discount factor, and the fact that the admissible surrender dates do not depend
on $\kappa_s$.
	
	\begin{prop}[Monotonicity with respect to the surrender-adjustment parameter]
		\label{prop:fee-monotone}
		Let $G>0$ and, for every fixed $c\geq0$, let
		$F^{(c)}=(F_t^{(c)})_{0\leq t\leq T}$ be an adapted, strictly positive
		fee-deducted account process that does not depend on $\kappa_s$. Assume that
		\[
		0<D_r(0,t)\leq \overline{D}<\infty,
		\qquad 0\leq t\leq T,
		\]
		and let $\mathcal T^{\mathcal D}_{0,T}$ be a nonempty set of admissible
		stopping times, independent of $\kappa_s$, with $T\in\mathcal D$.
		For $c\geq0$ and $\kappa_s\geq0$, define
		\[
		\Phi_{\kappa_s}(t)
		=
		\left(G-e^{-\kappa_s(T-t)}F_t^{(c)}\right)^+,
		\qquad 0\leq t\leq T,
		\]
		and
		\[
		U_{\mathrm{Sur}}(0;c,\kappa_s)
		=
		\sup_{\tau\in\mathcal T^{\mathcal D}_{0,T}}
		\E^{\Q}\!\left[D_r(0,\tau)\Phi_{\kappa_s}(\tau)\right].
		\]
		\begin{enumerate}
			\item[(i)] For every fixed $c\geq0$, the map
			\[
			\kappa_s\longmapsto U_{\mathrm{Sur}}(0;c,\kappa_s)
			\]
			is nondecreasing.
			
			\item[(ii)] Let $\kappa_2>\kappa_1$, and suppose that an optimal
			stopping time $\tau^*\in\mathcal T^{\mathcal D}_{0,T}$ exists for
			$\kappa_1$. If
			\begin{equation}
				\label{eq:nondeg}
				\Q(\mathcal E)>0,
				\qquad
				\mathcal E
				:=
				\left\{
				\tau^*<T,\;
				G-e^{-\kappa_1(T-\tau^*)}F_{\tau^*}^{(c)}>0
				\right\},
			\end{equation}
			then
			\[
			U_{\mathrm{Sur}}(0;c,\kappa_2)
			>
			U_{\mathrm{Sur}}(0;c,\kappa_1).
			\]
			
			\item[(iii)] Let
			\[
			H(c,\kappa_s)
			=
			\Gamma(c)+U_{\mathrm{Sur}}(0;c,\kappa_s)-F_0,
			\]
			where the account component $\Gamma(c)$ does not depend on $\kappa_s$.
			Suppose that, for every $\kappa_s$ under consideration, the map
			$c\mapsto H(c,\kappa_s)$ is continuous and strictly decreasing and
			admits a zero $c^*(\kappa_s)$. Then
			\[
			\kappa_s\longmapsto c^*(\kappa_s)
			\]
			is nondecreasing. Moreover, let $\kappa_2>\kappa_1$. If, at
			$c=c^*(\kappa_1)$, an optimal stopping time
			$\tau^*\in\mathcal T^{\mathcal D}_{0,T}$ exists for
			$U_{\mathrm{Sur}}(0;c,\kappa_1)$ and \eqref{eq:nondeg} holds for this
			$\tau^*$, then
			\[
			c^*(\kappa_2)>c^*(\kappa_1).
			\]
		\end{enumerate}
	\end{prop}
	
	The proof is given in Appendix~\ref{app:proof-fee-monotonicity}.

	In the LV and Heston SLV specifications considered in this paper, the account component in part~(iii) is
	\[
	\Gamma(c)=F_0e^{-cT}D_q(0,T),
	\]
	and the admissible surrender set is fixed independently of $\kappa_s$. Proposition~\ref{prop:fee-monotone} therefore applies directly to both models. In particular, under the stated monotonicity condition for the fair-fee equation, increasing the surrender-adjustment parameter cannot decrease the fair insurance fee. This contractual ordering will be used in Section~\ref{sec:numerics} to interpret the fair-fee calculations.

	\section{Hybrid forward calibration and backward valuation}
	\label{sec:hybrid}
	
	The local-volatility model is used as a one-factor benchmark and is valued by a one-dimensional finite-difference discretisation of the backward equation \eqref{eq:lv-backward}, with the surrender obstacle imposed at the admissible surrender dates when appropriate; the numerical specification is summarised in Subsection~\ref{app:lv-benchmark} of Appendix~\ref{app:numerical-details}. The main numerical construction concerns the Heston SLV model. It is organised in two successive stages. A forward hybrid pass calibrates the leverage surface to the prescribed local-volatility marginals. Once that surface has been stored on the full time--stock grid, a backward hybrid pass prices the GMMB by dynamic programming. The same recombining Heston tree is used in both stages, while conditional evolution in the remaining spatial coordinate is treated by one-dimensional PDE solves.
	
	The forward pass constitutes the core calibration stage of the numerical
	method and is therefore described first. In the correlated branch the mixed Brownian component is first removed by a change of variable; the state variable $V_t$ is then represented by a trinomial tree; conditional forward equations propagate the joint law; and the resulting density is projected onto the stock grid to evaluate the conditional mean of $V_t$ entering the Markovian-projection identity.
	
	\subsection{Decorrelation and transformed SDE}
	\label{subsec:rho-nonzero}
	
	For the correlated case $\rho\neq0$, let $W$ and $Z$ be independent
	Brownian motions and write
	\begin{equation}
		\dd W^V_t=\dd W_t,
		\qquad
		\dd W^S_t=\rho\,\dd W_t+\barrho\,\dd Z_t,
		\qquad
		\barrho=\sqrt{1-\rho^2}.
		\label{eq:brownian-decomp-main}
	\end{equation}
	Define
	\begin{equation}
		g(t,S)=\int_{S_0}^{S}\frac{\dd\xi}{\xi L(t,\xi)},
		\qquad
		Y_t=V_t-\frac{\omega}{\rho}g(t,S_t).
		\label{eq:method2-transform-main}
	\end{equation}
	Throughout this section, $g_t=\partial_t g$ and $L_S=\partial_S L$.
	Since
	\[
	g_S(t,S)=\frac{1}{S L(t,S)},
	\qquad
	g_{SS}(t,S)
	=
	-\frac{L(t,S)+S L_S(t,S)}
	{S^2L^2(t,S)},
	\]
	It\^o's formula cancels the Brownian component shared by $S$ and $V$
	and yields
	\begin{equation}
		\dd Y_t
		=
		\mu_Y(t,S_t,V_t)\dd t
		-\frac{\omega\barrho}{\rho}\sqrt{V_t}\dd Z_t,
		\label{eq:Y-dynamics-main}
	\end{equation}
	where
	\begin{equation}
		\begin{aligned}
			\mu_Y(t,S,V)
			=&\;
			\kappa_V(\theta_V-V)
			-\frac{\omega}{\rho}g_t(t,S)
			-\frac{\omega}{\rho}
			\frac{r_{t}-q_{t}}{L(t,S)}
			\\
			&+
			\frac{\omega}{2\rho}
			\left(L(t,S)+S L_S(t,S)\right)V.
		\end{aligned}
		\label{eq:muY-main}
	\end{equation}
	For a fixed value $v$ of $V_t$, define
	\begin{equation}
		\nu_Y(v)
		=
		\frac12
		\frac{\omega^2(1-\rho^2)}{\rho^2}v.
		\label{eq:nuY-main}
	\end{equation}
	
	The stock level associated with a transformed state $(y,v)$ is
	characterised by
	\begin{equation}
		g(t,S)
		=
		\frac{\rho}{\omega}(v-y).
		\label{eq:inverse-transform-main}
	\end{equation}
	Whenever this equation has a unique solution on the numerical stock
	domain, we denote it by $S(t,y,v)$.
	
	Let $p(t,y,v)$ denote the joint density of $(Y_t,V_t)$ with respect to
	the transformed coordinates $(y,v)$. The transformation removes the
	mixed second-order derivative from the joint forward equation. In the
	hybrid splitting, the evolution of $V_t$ is represented by the
	recombining tree described in Section~\ref{subsec:variance-tree}, while,
	during the spatial substep and for a fixed value $v$ of $V_t$, the
	transformed-coordinate component is advanced by
	\begin{equation}
		\partial_t p(t,y;v)
		=
		-\partial_y\!\left(
		\mu_Y\bigl(t,S(t,y,v),v\bigr)p(t,y;v)
		\right)
		+\nu_Y(v)\partial_{yy}p(t,y;v).
		\label{eq:forward-rho-nonzero-main}
	\end{equation}
	
	When $\rho=0$, no decorrelating transformation is required. We use
	\begin{equation}
		x=\log S
		\label{eq:log-coordinate-main}
	\end{equation}
	and, for a fixed value $v$ of $V_t$, define
	\begin{equation}
		\nu(t,x;v)
		=
		\frac12L^2(t,e^x)v,
		\qquad
		\mu(t,x;v)
		=
		r_{t}-q_{t}-\nu(t,x;v).
		\label{eq:rho0-coefficients-main}
	\end{equation}
	The evolution of $V_t$ is again represented by the CIR tree, while,
	during the spatial substep and for a fixed value $v$ of $V_t$, the
	log-stock component is advanced directly by
	\begin{equation}
		\partial_t p
		=
		-\partial_x(\mu p)
		+\partial_{xx}(\nu p).
		\label{eq:forward-rho0-main}
	\end{equation}
	
	Hence, in both branches, the hybrid splitting reduces the continuous
	spatial part of the forward calibration problem to a one-dimensional
	transport--diffusion equation coupled with the discrete Heston tree
	evolution.

	\subsection{Recombining trinomial tree for the Heston state variable}
\label{subsec:variance-tree}
\label{subsec:heston-tree}

Only the abstract ingredients of the variance discretisation are needed in the
main text. The time grid, the deterministic rate increments and the Heston tree
are constructed before the forward and backward sweeps. The CIR state variable
$V_t$ is approximated by a time-homogeneous recombining trinomial tree rooted at
$V_0$. We denote the variance nodes by $v_j$, the three children of a parent
node $v_j$ by the index set $\mathcal C(j)$, and the corresponding transition
probabilities by $\pi_{j\to h}$, $h\in\mathcal C(j)$. The tree matches the first
two exact conditional moments of the CIR process over one time interval. Its
child maps and transition probabilities are stored once and are then used in
both directions: they redistribute probability mass during the forward
calibration and form conditional expectations during the backward valuation.
The square-root lattice, moment-matching equations, admissibility search and
finite-tree treatment are described in Subsection~\ref{app:trinomial-tree} of Appendix~\ref{app:numerical-details}.

\subsection{Forward calibration of the leverage surface}
\label{subsec:forward-leverage}

The forward calibration follows the chronological order of the numerical
implementation, while the detailed formulas are deferred to
Appendix~\ref{app:numerical-details}. Let $S_\ell$, $\ell=0,\ldots,N_S$, be the
fixed stock grid on which the target local volatility and the leverage surface
are stored. In the correlated branch, let $y_m$, $m=0,\ldots,N_Y$, denote the
transformed grid, let $\mathcal I_Y$ denote linear interpolation on that grid,
and let $P^i_{j,m}$ represent the forward density-like array at time $t_i$,
variance node $v_j$ and transformed node $y_m$. We denote by
$\varepsilon_V>0$ the positive floor used in the conditional-variance and
leverage updates. Since
$g(0,S_0)=0$, the transformed initial state is $Y_0=V_0$; the initial mass is
therefore placed at the Heston-tree root and at the closest interior
$Y$-grid node. The leverage slice used to start the first time interval is
extended over the stock grid according to
\[
L(0,S_\ell)
=
\frac{\sigLV(0,S_\ell)}
{\sqrt{\max\{V_0,\varepsilon_V\}}},
\qquad \ell=0,\ldots,N_S.
\]

Assume that the accepted density $P^i$ and leverage slice $L^i$ are available
at time $t_i$. The leverage slice first determines the auxiliary stock-grid
arrays $g^i$, $g_t^i$ and $L^i+S L_S^i$. At each pair $(y_m,v_j)$, the inverse
relation \eqref{eq:inverse-transform-main} is then used to recover the stock
value $S(t_i,y_m,v_j)$. Interpolating the leverage-dependent arrays at these
recovered stock values provides the drift in
\eqref{eq:forward-rho-nonzero-main}; the diffusion coefficient is fixed by the
current variance node. Thus the accepted leverage information at $t_i$
determines all coefficients needed to propagate the density over
$[t_i,t_{i+1}]$.

Conditional on a node $v_j$, the one-dimensional transport--diffusion equation
is advanced by a conservative fully implicit finite-volume scheme. The interval
may be divided into equal implicit substeps according to the local drift
indicator described in Subsection~\ref{app:forward-calibration} of Appendix~\ref{app:numerical-details}. The drift and diffusion coefficients remain
frozen during these substeps. Denoting by $\widetilde P^{i+1}_{j,m}$ the output
of this factorwise solve before the tree transition, the Heston-tree
redistribution is
\begin{equation}
P^{i+1}_{h,m}
=
\sum_{j:\,h\in\mathcal C(j)}
\pi_{j\to h}\,\widetilde P^{i+1}_{j,m}.
\label{eq:forward-tree-redistribution-main}
\end{equation}
This produces the joint density at the new time level. The forward PDE solve
and the tree redistribution are performed once over each time interval.

The new leverage slice is then obtained by projecting the propagated density
onto the fixed stock grid. Given a provisional new-time map
$g^{i+1,[n]}$ at iteration $n$, the transformed coordinate associated with the
pair $(S_\ell,v_j)$ is
\[
y_j^{i+1,[n]}(S_\ell)
=
v_j-\frac{\omega}{\rho}g^{i+1,[n]}(S_\ell).
\]
The density is evaluated at this point by interpolation on the transformed
grid. The resulting approximation of the conditional mean is
\begin{equation}
\widehat v^{i+1,[n]}(S_\ell)
=
\frac{
\sum_j v_j\,
\mathcal I_Y[P^{i+1}_{j,\cdot}]
\bigl(y_j^{i+1,[n]}(S_\ell)\bigr)
}{
\sum_j
\mathcal I_Y[P^{i+1}_{j,\cdot}]
\bigl(y_j^{i+1,[n]}(S_\ell)\bigr)
}.
\label{eq:conditional-variance-main}
\end{equation}
The Markovian-projection identity gives the direct update
\begin{equation}
\widehat L^{i+1,[n+1]}(S_\ell)
=
\frac{\sigLV(t_{i+1},S_\ell)}
{\sqrt{\max\{\widehat v^{i+1,[n]}(S_\ell),\varepsilon_V\}}}.
\label{eq:discrete-leverage-update}
\end{equation}
The first projection uses $g^{i+1,[0]}:=g^i$ and accepts the corresponding
direct update as $L^{i+1,[1]}=\widehat L^{i+1,[1]}$. The map and the auxiliary
arrays are then rebuilt. For subsequent projection passes, the propagated
density $P^{i+1}$ remains fixed and only the new-time map, conditional moment
and leverage slice are updated. With relaxation parameter $\alpha\in(0,1]$,
the iteration is
\[
L^{i+1,[n+1]}
=
\alpha\widehat L^{i+1,[n+1]}
+(1-\alpha)L^{i+1,[n]},
\qquad n\geq1.
\]
Once the last projection pass has been accepted, the final leverage slice and
its auxiliary arrays are stored and the algorithm advances to the next time
interval. Repeating this procedure up to maturity produces the calibrated
leverage surface on the full time--stock grid.

When $\rho=0$, the same calibration logic is implemented directly on the
log-stock grid. The construction of $g$ and $g_t$, the inversion $Y\mapsto S$
and the interpolation back to the stock grid are not needed. After the
factorwise log-stock solve and the same Heston-tree redistribution, the
conditional mean is read directly from the propagated array,
\[
\widehat v^{i+1}(S_m)
=
\frac{\sum_j v_jP^{i+1}_{j,m}}
{\sum_j P^{i+1}_{j,m}},
\qquad S_m=e^{x_m}.
\]
The leverage update is then identical to
\eqref{eq:discrete-leverage-update}. Subsections~\ref{app:spatial-grids}--\ref{app:independent-branch} of Appendix~\ref{app:numerical-details} give the sequential construction of the grids, the
leverage-dependent arrays, the inverse map, the finite-volume coefficients,
the substep rule, the low-density fallback and the projection iteration.

\subsection{Backward pricing on the calibrated leverage surface}
\label{subsec:backward-step}

Once the forward calibration has reached maturity, the leverage surface and all
associated map and derivative arrays are frozen. The backward sweep uses the
same Heston tree but traverses the splitting in reverse order. For the
correlated branch, the continuation equation conditional on $V_t=v$ is
\begin{equation}
\partial_t u
+\mu_Y\bigl(t,S(t,y,v),v\bigr)\partial_yu
+\nu_Y(v)\partial_{yy}u-r_tu=0.
\label{eq:backward-rho-nonzero-main}
\end{equation}
The value array is initialised at maturity by
\begin{equation}
U^{N_t}_{j,m}
=
\left(
G-F_0\frac{S(T,y_m,v_j)}{S_0}e^{-cT}
\right)^+.
\label{eq:terminal-discrete}
\end{equation}

At a generic backward time level, the values at the three child variance nodes
are first averaged with the stored tree probabilities,
\begin{equation}
R^i_{j,m}
=
\sum_{h\in\mathcal C(j)}
\pi_{j\to h}U^{i+1}_{h,m}.
\label{eq:tree-expectation-main}
\end{equation}
This vector provides the terminal datum for the conditional spatial solve over
$[t_i,t_{i+1}]$. The stored leverage slice at $t_i$ is used to recover the
stock values represented by the spatial grid and to evaluate the drift and
diffusion coefficients. As in the forward calculation, the time interval may
be split into equal implicit substeps according to the local drift indicator.
The corresponding financial boundary values are imposed in the recovered stock
coordinate, with different low-stock conditions for the terminal-only and
surrenderable guarantees as detailed in Subsection~\ref{app:backward-pricing} of Appendix~\ref{app:numerical-details}.

Let $\mathcal P^{\rm FD}_{i,j}$ denote the resulting undiscounted implicit
finite-difference propagation operator at node $v_j$. After the conditional
spatial solve, the exact deterministic discount factor is applied, so that
\begin{equation}
\widetilde U^i_{j,\cdot}
=
D_r(t_i,t_{i+1})\,
\mathcal P^{\rm FD}_{i,j}\!\left(R^i_{j,\cdot}\right).
\label{eq:backward-continuation}
\end{equation}
At an admissible surrender date, the continuation value is then projected
pointwise onto the immediate surrender payoff,
\begin{equation}
U^i_{j,m}
=
\max\left\{
\widetilde U^i_{j,m},
\left(
G-e^{-\kappa_s(T-t_i)}F_0
\frac{S(t_i,y_m,v_j)}{S_0}e^{-ct_i}
\right)^+
\right\}.
\label{eq:discrete-obstacle}
\end{equation}
For a guarantee without surrender, this projection is omitted. Under the
reported discrete-time convention, it is imposed at every positive
pre-maturity grid date and at maturity through the terminal payoff, but not at
$t_0=0$. The resulting value array becomes the input to the preceding backward
time level. Repetition down to time zero completes the valuation and, by
locating the transition between equality with the obstacle and strict
continuation at each variance node, also produces the state-dependent surrender
boundary.

For $\rho=0$, the same tree/finite-difference recursion is performed directly
in $x=\log S$, with $S=e^x$ and the coefficients in
\eqref{eq:rho0-coefficients-main}. Subsection~\ref{app:backward-pricing} of Appendix~\ref{app:numerical-details} gives the detailed order of the tree expectation, coefficient and
substep selection, boundary treatment, implicit solve, discounting and obstacle
projection. Because the leverage calibration is a nonlinear projection of a
numerically represented conditional law, convergence is assessed in
Section~\ref{sec:numerics} by recalibrating the leverage surface on increasingly
refined grids before comparing the resulting GMMB values.

\section{Numerical results}
	\label{sec:numerics}
	
	This section reports two numerical experiments for the guaranteed minimum maturity benefit with surrender option. The contract is valued as the guarantee-rider component only, consistently with the decomposition of the variable annuity into the account component and the embedded put-like guarantee. Accordingly, ``surrender'' below refers to surrender of this rider and not to lapse of the entire account. We first compare the local-volatility model with the SLV model equipped with its calibrated leverage function in a controlled synthetic setting, using the leverage surface produced by the forward calibration pass of the hybrid method. The contract, market, Heston SLV and numerical parameters for both experiments are reported in Table~\ref{tab:numerical-inputs}, and Figure~\ref{fig:slv-leverage-function} displays the synthetic local-volatility surface and the corresponding calibrated SLV leverage function. Subsection~\ref{subsec:real-market-eurostoxx} then repeats the comparison using EURO STOXX 50 market data. Recall that $N_t$ denotes the number of time intervals. Let $N_S$ denote the number of log-spot intervals used for the LV benchmark and leverage grid, and let $N_Y$ denote the number of intervals on the transformed SLV PDE grid. For the reported correlated calculations we set $N_Y=N_S$. In the convergence study, $N_{\mathrm{grid}}$ denotes the common refinement level
	\[
	N_{\mathrm{grid}}:=N_t=N_S=N_Y.
	\]
	
	\begin{table}[H]
		\centering
		\scriptsize
		\renewcommand{\arraystretch}{1.08}
		\setlength{\tabcolsep}{4pt}
		\begin{tabularx}{\textwidth}{@{}l >{\raggedright\arraybackslash}p{0.28\textwidth} >{\raggedright\arraybackslash}X >{\raggedright\arraybackslash}X@{}}
			\toprule
			Symbol & Meaning & Synthetic data & Market data \\
			\midrule
			\multicolumn{4}{@{}l}{\textit{Contract parameters}} \\
			$G$ & Guarantee level & $100$ & $100$ \\
			$T$ & Maturity & $10$ years & $10$ years \\
			$F_0$ & Initial account values & $40,50,\ldots,160$ & $40,50,\ldots,160$ \\
			$(c,\kappa_s)$ & Fee--surrender-adjustment pairs & $(0,0)$, $(0.01,0.01)$, $(0.03,0)$, $(0.03,0.02)$, $(0.03,0.03)$ & same \\
			-- & Surrender convention & without surrender; discrete-time optimal surrender & same \\
			\addlinespace[0.5ex]
			\multicolumn{4}{@{}l}{\textit{Market and Heston SLV parameters}} \\
			$S_0$ & Reference fund/index level & $1000$ & $6265.58$ \\
			$r_t$ & Risk-free rate & $r_t\equiv0.04$ & deterministic market-implied curve; see Fig.~\ref{fig:real-market-rq-curves} \\
			$q_t$ & dividend yield & $q_t\equiv0$ & deterministic market-implied curve; see Fig.~\ref{fig:real-market-rq-curves} \\
			$\sigLV$ & Target local volatility & synthetic non-flat surface, Eq.~\eqref{eq:wang-synthetic-lv} & calibrated EURO STOXX 50 surface \\
			$V_0$ & Initial level of $V_t$ & $0.04$ & $0.0265$ \\
			$\kappa_V$ & Mean reversion & $2.0$ & $3.70$ \\
			$\theta_V$ & Long-run level of $V_t$ & $0.04$ & $0.0461$ \\
			$\omega$ & Diffusion coefficient of $V_t$ & $0.30$ & $0.50$ \\
			$\rho$ & Correlation & $-0.70$ & $-0.75$ \\
			\addlinespace[0.5ex]
			\multicolumn{4}{@{}l}{\textit{Numerical specification}} \\
			$L$ & Leverage surface & forward SLV calibration & forward SLV calibration \\
			$N_Y,N_t$ & Main hybrid PDE grid & $2000,2000$ & $4000,4000$ \\
			$N_S,N_t$ & Main LV PDE/leverage grid & $2000,2000$ & $4000,4000$ \\
			$N_{\mathrm{grid}}$ & Common convergence level & $250,500,1000,2000,4000$ & main level $4000$ \\
			\bottomrule
		\end{tabularx}
		\normalsize
		\caption{Contract, market, Heston SLV and numerical parameters for the synthetic and market-data experiments.}
		\label{tab:numerical-inputs}
	\end{table}
	
	The fee parameter $c$ is the continuously compounded fee entering the fee-deducted policy-account value $F_t^{(c)}$. The surrender-adjustment parameter is denoted by $\kappa_s$ throughout the numerical section and is the rate in the exponential adjustment $e^{-\kappa_s(T-t)}$, rather than a one-off percentage charge. The cases labelled ``with surrender'' use a discrete-time optimal-surrender convention: the obstacle is imposed at every positive backward grid date and at maturity, but not at $t=0$. The reported values are therefore grid-based approximations to the corresponding continuous-time optimal-surrender problem. The surrender boundaries are reported in the policy-account coordinate $F$, while $S$ remains reserved for the reference fund or index before fee deduction.
	
	Numerical tolerances, interpolation rules and substepping are common to the forward and backward implementations described in Appendix~\ref{app:numerical-details}; experiment-specific values are stated below when they enter the reported calculations.

	\subsection{Synthetic Data}
	\label{subsec:synthetic-data}
	
	\begin{figure}[H]
		\centering
		\IfFileExists{synthetic_LV_leverage_surfaces.png}{%
			\includegraphics[width=0.90\textwidth]{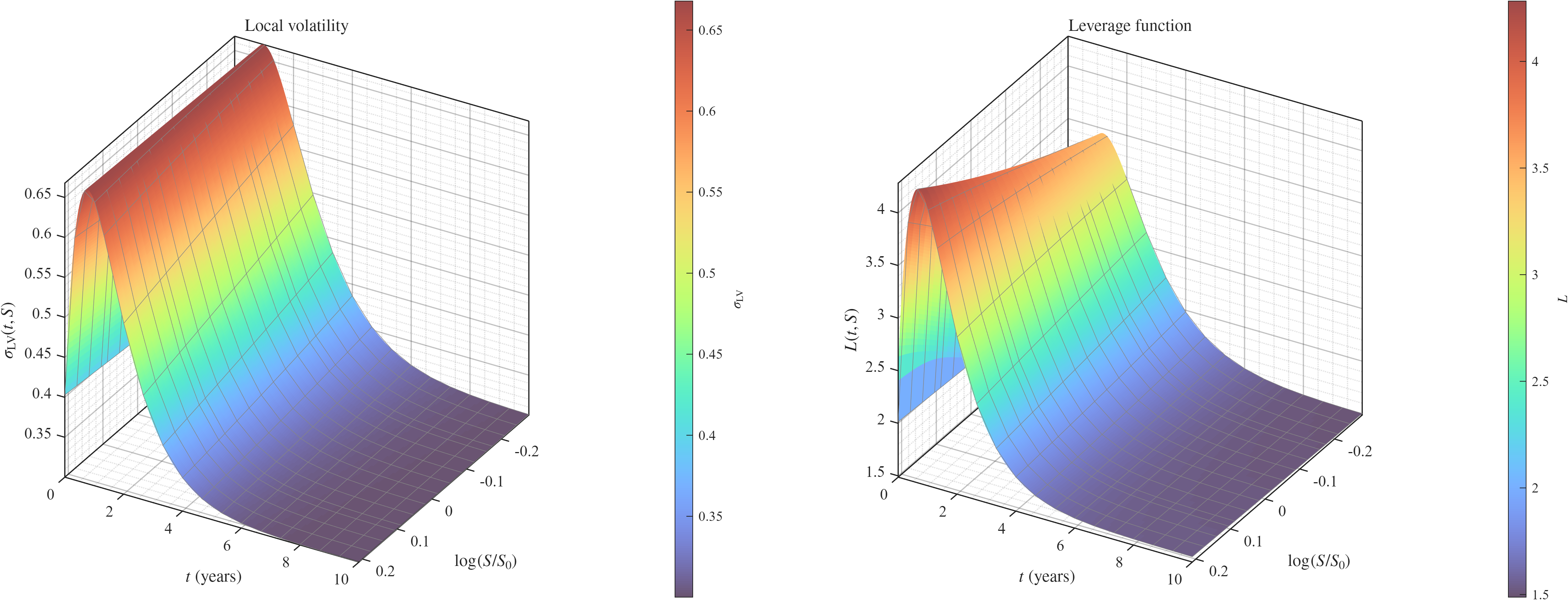}%
		}{%
			\fbox{\parbox[c][0.30\textheight][c]{0.95\textwidth}{%
					\centering Figure file
					\texttt{synthetic\_LV\_leverage\_surfaces} not found.}}%
		}
		\caption{Synthetic local-volatility surface (left) and corresponding
			calibrated SLV leverage function (right), shown against maturity and
			log-moneyness $\log(S/S_0)$.}
		\label{fig:slv-leverage-function}
	\end{figure}
	
	The numerical tests use a synthetic, non-flat local-volatility surface that is not calibrated from market option data. It is a fixed, controlled input for comparing LV and SLV under the same projected local-volatility structure, thereby isolating the effect of stochastic volatility conditional on that input. Following the synthetic example of \citet{Wang2025}, the target local-volatility surface is
	\begin{equation} 
		\sigLV(t,S)=0.3+\psi e^{-\psi},\qquad \psi=(t+0.1)\sqrt{\frac{S}{S_0}+0.1}.
		\label{eq:wang-synthetic-lv}
	\end{equation}
	The dimensionless ratio $S/S_0$ avoids making the synthetic surface depend on the arbitrary scale of the initial fund level. The same target surface is used for the one-factor LV benchmark and for the SLV calibration. The objective is to isolate the residual effect of stochastic volatility when both LV and SLV are built from the same non-flat volatility structure. The reported results therefore compare GMMB guarantee values without surrender and with a surrender option under LV and SLV.

	{
		\renewcommand{\PriceDecimals}{4}
		
		We first study the numerical convergence of the deterministic pricing schemes on
		the representative contract
		\[
		G=F_0=100,\qquad T=10,\qquad c=3\%,\qquad \kappa_s=2\%.
		\]
		The market and Heston SLV parameters are those reported in Table~\ref{tab:numerical-inputs}.
		For each refinement level \(N_{\mathrm{grid}}\), the SLV leverage function is
		recalibrated up to maturity \(T=10\) on the same time grid and is then used
		in the backward hybrid tree/finite-difference valuation step. The LV
		benchmark is computed by a one-dimensional finite-difference PDE using the
		same number of spatial and temporal nodes. The times shown in
		Tables~\ref{tab:convergence-european-gmmb} and
		\ref{tab:convergence-surrender-gmmb} refer only to the corresponding
		backward valuation steps. At each refinement level, the leverage surface is
		calibrated once in a separate forward pass and is then reused by both the
		terminal-only and surrender valuations. The corresponding forward
		leverage-calibration CPU times are \(0.014\), \(0.095\), \(0.742\),
		\(6.048\) and \(54.058\) seconds for
		\(N_{\mathrm{grid}}=250,500,1000,2000\) and \(4000\), respectively.
		These one-off calibration costs are therefore excluded from the pricing
		times reported in the tables.
		
		Table~\ref{tab:convergence-european-gmmb} reports the GMMB results without surrender.
		As expected, the LV finite-difference values converge very rapidly. The SLV
		hybrid values display a slightly less monotone behaviour on the coarser grids,
		because they combine the error of the backward solver with the interpolation
		error of the calibrated leverage function. Nevertheless, the SLV price
		stabilises around the LV value as the grid is refined. On the finest grid,
		\(N_{\mathrm{grid}}=4000\), the SLV--LV difference is only \(\pnum{0.0028}\), which is
		negligible on the scale of the contract. This is consistent with the fact that the terminal payoff depends only on the terminal fee-deducted policy-account value and therefore mainly tests the marginal calibration of the SLV model to the target local-volatility surface. The finest deterministic SLV value is close to, but slightly below, the reported Monte Carlo (MC) 95\% interval. We therefore use the simulation as an independent consistency check rather than as a formal statistical validation of the deterministic value.
		
		\begin{table}[htbp]
			\centering
			
			\begin{adjustbox}{max width=\textwidth}
				{
					\renewcommand{\arraystretch}{1.35}
					\begin{tabular}{rrrr}
						\toprule
						\(N_{\mathrm{grid}}\) & LV PDE & SLV hybrid & SLV MC \\
						\midrule
						250
						& \(\substack{\pnum{30.6293}\\(0.0020\ \mathrm{s})}\)
						& \(\substack{\pnum{30.2915}\\(0.045\ \mathrm{s})}\)
						& \multirow{5}{*}{\(\substack{\pnum{30.6561}\\\pm\pnum{0.0225}}\)} \\
						500
						& \(\substack{\pnum{30.6286}\\(0.0060\ \mathrm{s})}\)
						& \(\substack{\pnum{30.4801}\\(0.23\ \mathrm{s})}\)
						& \\
						1000
						& \(\substack{\pnum{30.6284}\\(0.025\ \mathrm{s})}\)
						& \(\substack{\pnum{30.6801}\\(1.5\ \mathrm{s})}\)
						& \\
						2000
						& \(\substack{\pnum{30.6284}\\(0.13\ \mathrm{s})}\)
						& \(\substack{\pnum{30.6243}\\(11\ \mathrm{s})}\)
						& \\
						4000
						& \(\substack{\pnum{30.6283}\\(0.40\ \mathrm{s})}\)
						& \(\substack{\pnum{30.6311}\\(94\ \mathrm{s})}\)
						& \\
						\bottomrule
					\end{tabular}
				}
			\end{adjustbox}
			\normalsize
			\caption{Convergence of European GMMB guarantee prices without surrender under LV and SLV.
				The SLV Monte Carlo benchmark uses \(5\cdot10^6\) paths and the leverage function
				calibrated on the finest grid. Monte Carlo values are reported as
				\(x\pm1.96\,\mathrm{SE}\). Deterministic times refer to the backward valuation only;
				forward leverage-calibration times are reported in the text.}
			\label{tab:convergence-european-gmmb}
		\end{table}
		
		Table~\ref{tab:convergence-surrender-gmmb} gives the corresponding results for
		the GMMB with discrete-time optimal surrender, with surrender allowed at every positive time-grid date. The LV finite-difference values
		again converge smoothly. The SLV hybrid prices stabilise around \(43.2\),
		and the SLV premium relative to LV remains positive across all refinement
		levels. On the finest grid, the premium is \(\pnum{0.8726}\), corresponding to
		roughly \(2.1\%\) of the LV price. This persistence under refinement supports the conclusion that the SLV--LV difference is not generated solely by the displayed spatial and temporal discretisation. The independent-policy least-squares Monte Carlo (LSMC) estimate is lower than the finest deterministic value, as is typical for a value obtained by applying an estimated stopping policy; its reported sampling interval does not quantify that policy bias. The comparison is therefore interpreted as an order-of-magnitude and policy-consistency check.
		
		\begin{table}[htbp]
			\centering 
			
			\begin{adjustbox}{max width=\textwidth}
				{
					\renewcommand{\arraystretch}{1.35}
					\begin{tabular}{rrrr}
						\toprule
						\(N_{\mathrm{grid}}\) & LV PDE & SLV hybrid & SLV LSMC \\
						\midrule
						250
						& \(\substack{\pnum{42.3120}\\(\ptime{0.0030}\ \mathrm{s})}\)
						& \(\substack{\pnum{42.6312}\\(\ptime{0.0610}\ \mathrm{s})}\)
						& \multirow{5}{*}{\(\substack{\pnum{43.0923}\\\pm\pnum{0.0519}}\)} \\
						500
						& \(\substack{\pnum{42.3234}\\(\ptime{0.0140}\ \mathrm{s})}\)
						& \(\substack{\pnum{42.9243}\\(\ptime{0.3680}\ \mathrm{s})}\)
						& \\
						1000
						& \(\substack{\pnum{42.3291}\\(\ptime{0.0490}\ \mathrm{s})}\)
						& \(\substack{\pnum{43.2502}\\(\ptime{2.4500}\ \mathrm{s})}\)
						& \\
						2000
						& \(\substack{\pnum{42.3321}\\(\ptime{0.2070}\ \mathrm{s})}\)
						& \(\substack{\pnum{43.1852}\\(\ptime{18.8400}\ \mathrm{s})}\)
						& \\
						4000
						& \(\substack{\pnum{42.3335}\\(\ptime{0.7740}\ \mathrm{s})}\)
						& \(\substack{\pnum{43.2061}\\(\ptime{155.7440}\ \mathrm{s})}\)
						& \\
						\bottomrule
					\end{tabular}
					
				}
			\end{adjustbox}
			\normalsize
			\caption{Convergence of GMMB prices with discrete-time optimal surrender under LV and SLV.
				The SLV LSMC benchmark uses \(5\cdot10^6\) paths, with \(4\cdot10^6\) training paths
				and \(10^6\) independent pricing paths, on the \(4000\)-date surrender grid.
				LSMC values are reported as \(x\pm1.96\,\mathrm{SE}\). Deterministic times refer to
				the backward valuation only; forward leverage-calibration times are reported in the text.}
			\label{tab:convergence-surrender-gmmb}
		\end{table}
		
		The computational times show the different complexity of the two deterministic
		pricing procedures. The one-dimensional LV PDE remains very fast even on the
		finest grid. The SLV hybrid method is more expensive because the backward
		valuation is coupled with the Heston tree and uses the calibrated leverage
		surface. The cost increases significantly on the finest levels, but the
		\(N_{\mathrm{grid}}=2000\) grid already gives a stable estimate of the SLV
		premium while requiring a considerably smaller computational time than the
		\(4000\)-level run. The forward calibration constitutes a separate one-off
		preprocessing stage. Its CPU time increases from \(0.014\) seconds at
		\(N_{\mathrm{grid}}=250\) to \(6.048\) seconds at
		\(N_{\mathrm{grid}}=2000\) and \(54.058\) seconds at
		\(N_{\mathrm{grid}}=4000\). This increase reflects both the propagation of
		the joint law over the full time horizon and the increasing number of spatial,
		temporal and Heston-tree nodes involved in the calibration. The stored
		leverage surface itself contains approximately
		\((N_t+1)(N_S+1)\) values.

		Overall, the results without surrender support the numerical marginal consistency of the SLV calibration, whereas the results with surrender show a stable and economically visible difference once early surrender is introduced. The Monte Carlo and independent-policy LSMC calculations provide useful independent checks, subject to the qualifications stated above.
		
	}

	\medskip
	\noindent
	We now turn to the full pricing tables. The following results use the main pricing grids reported in Table~\ref{tab:numerical-inputs} and compare guarantee values without surrender and with a surrender option across several initial account values and fee--surrender-adjustment combinations.

	\begin{table}[H]
		\centering
		\renewcommand{\arraystretch}{1.06}
		
		\resizebox{\textwidth}{!}{%
			\begin{tabular}{rccccccccccccccc}
				\toprule
				&
				&
				\multicolumn{2}{c}{$c=0\%,\,\kappa_s=0\%$}
				&
				&
				\multicolumn{2}{c}{$c=1\%,\,\kappa_s=1\%$}
				&
				&
				\multicolumn{2}{c}{$c=3\%,\,\kappa_s=0\%$}
				&
				&
				\multicolumn{2}{c}{$c=3\%,\,\kappa_s=2\%$}
				&
				&
				\multicolumn{2}{c}{$c=3\%,\,\kappa_s=3\%$} \\
				
				\cmidrule(lr){3-4}
				\cmidrule(lr){6-7}
				\cmidrule(lr){9-10}
				\cmidrule(lr){12-13}
				\cmidrule(lr){15-16}
				
				$F_0$
				&
				& LV & SLV
				&
				& LV & SLV
				&
				& LV & SLV
				&
				& LV & SLV
				&
				& LV & SLV \\
				\midrule
				
				40
				&
				& \pnum{40.4389} & \pnum{40.4665}
				&
				& \pnum{41.9923} & \pnum{41.9638}
				&
				& \pnum{45.0021} & \pnum{44.9855}
				&
				& \pnum{45.0021} & \pnum{44.9855}
				&
				& \pnum{45.0021} & \pnum{44.9855} \\
				
				50
				&
				& \pnum{36.8996} & \pnum{36.8969}
				&
				& \pnum{38.4950} & \pnum{38.4826}
				&
				& \pnum{41.6352} & \pnum{41.6497}
				&
				& \pnum{41.6352} & \pnum{41.6497}
				&
				& \pnum{41.6352} & \pnum{41.6497} \\
				
				60
				&
				& \pnum{33.9780} & \pnum{33.9377}
				&
				& \pnum{35.5807} & \pnum{35.5791}
				&
				& \pnum{38.7757} & \pnum{38.7443}
				&
				& \pnum{38.7757} & \pnum{38.7443}
				&
				& \pnum{38.7757} & \pnum{38.7443} \\
				
				70
				&
				& \pnum{31.5224} & \pnum{31.5556}
				&
				& \pnum{33.1124} & \pnum{33.0835}
				&
				& \pnum{36.3156} & \pnum{36.2852}
				&
				& \pnum{36.3156} & \pnum{36.2852}
				&
				& \pnum{36.3156} & \pnum{36.2852} \\
				
				80
				&
				& \pnum{29.4266} & \pnum{29.4464}
				&
				& \pnum{30.9927} & \pnum{31.0021}
				&
				& \pnum{34.1752} & \pnum{34.1397}
				&
				& \pnum{34.1752} & \pnum{34.1397}
				&
				& \pnum{34.1752} & \pnum{34.1397} \\
				
				90
				&
				& \pnum{27.6144} & \pnum{27.6291}
				&
				& \pnum{29.1506} & \pnum{29.1333}
				&
				& \pnum{32.2948} & \pnum{32.2656}
				&
				& \pnum{32.2948} & \pnum{32.2656}
				&
				& \pnum{32.2948} & \pnum{32.2656} \\
				
				100
				&
				& \pnum{26.0290} & \pnum{26.0591}
				&
				& \pnum{27.5329} & \pnum{27.5789}
				&
				& \pnum{30.6284} & \pnum{30.6278}
				&
				& \pnum{30.6284} & \pnum{30.6278}
				&
				& \pnum{30.6284} & \pnum{30.6278} \\
				
				110
				&
				& \pnum{24.6522} & \pnum{24.6911}
				&
				& \pnum{26.0988} & \pnum{26.1017}
				&
				& \pnum{29.1402} & \pnum{29.1268}
				&
				& \pnum{29.1402} & \pnum{29.1268}
				&
				& \pnum{29.1402} & \pnum{29.1268} \\
				
				120
				&
				& \pnum{23.4287} & \pnum{23.4299}
				&
				& \pnum{24.8376} & \pnum{24.8698}
				&
				& \pnum{27.8021} & \pnum{27.8129}
				&
				& \pnum{27.8021} & \pnum{27.8129}
				&
				& \pnum{27.8021} & \pnum{27.8129} \\
				
				130
				&
				& \pnum{22.3341} & \pnum{22.3221}
				&
				& \pnum{23.7063} & \pnum{23.7265}
				&
				& \pnum{26.5912} & \pnum{26.5991}
				&
				& \pnum{26.5912} & \pnum{26.5991}
				&
				& \pnum{26.5912} & \pnum{26.5991} \\
				
				140
				&
				& \pnum{21.3489} & \pnum{21.3516}
				&
				& \pnum{22.6848} & \pnum{22.7221}
				&
				& \pnum{25.4979} & \pnum{25.5356}
				&
				& \pnum{25.4979} & \pnum{25.5356}
				&
				& \pnum{25.4979} & \pnum{25.5356} \\
				
				150
				&
				& \pnum{20.4572} & \pnum{20.4475}
				&
				& \pnum{21.7577} & \pnum{21.7838}
				&
				& \pnum{24.5077} & \pnum{24.5379}
				&
				& \pnum{24.5077} & \pnum{24.5379}
				&
				& \pnum{24.5077} & \pnum{24.5379} \\
				
				160
				&
				& \pnum{19.6460} & \pnum{19.6582}
				&
				& \pnum{20.9124} & \pnum{20.9043}
				&
				& \pnum{23.5998} & \pnum{23.5981}
				&
				& \pnum{23.5998} & \pnum{23.5981}
				&
				& \pnum{23.5998} & \pnum{23.5981} \\
				
				\bottomrule
			\end{tabular}%
		}
		
		\caption{European GMMB guarantee prices in the LV and SLV models.}
		\label{tab:european-lv-slv-results}
	\end{table}

	Table~\ref{tab:european-lv-slv-results} confirms the expected near-coincidence
	between LV and SLV prices when the payoff depends only on the terminal
	fee-deducted policy-account value. Across all entries in the table, the maximum relative SLV--LV price
	discrepancy, computed relative to the LV price, is approximately
	$0.167\%$, while the maximum absolute discrepancy is about $0.046$
	(both based on the unrounded numerical values underlying
	Table~\ref{tab:european-lv-slv-results}). Both maxima are attained at
	$F_0=100$ for $(c,\kappa_s)=(1\%,1\%)$. These residual differences are
	consistent with numerical pricing error, interpolation of the calibrated
	leverage surface, and the finite calibration grid. Notice also that, without
	surrender, the three columns with $c=3\%$ are identical across $\kappa_s$, as
	expected, since the surrender-adjustment parameter $\kappa_s$ does not enter
	the terminal-only payoff.

	\begin{table}[H]
		{\centering
			\renewcommand{\arraystretch}{1.08}
			
			\resizebox{\textwidth}{!}{%
				\begin{tabular}{rccccccccccccccc}
					\toprule
					&
					&
					\multicolumn{2}{c}{$c=0\%,\,\kappa_s=0\%$}
					&
					&
					\multicolumn{2}{c}{$c=1\%,\,\kappa_s=1\%$}
					&
					&
					\multicolumn{2}{c}{$c=3\%,\,\kappa_s=0\%$}
					&
					&
					\multicolumn{2}{c}{$c=3\%,\,\kappa_s=2\%$}
					&
					&
					\multicolumn{2}{c}{$c=3\%,\,\kappa_s=3\%$} \\
					
					\cmidrule(lr){3-4}
					\cmidrule(lr){6-7}
					\cmidrule(lr){9-10}
					\cmidrule(lr){12-13}
					\cmidrule(lr){15-16}
					
					$F_0$
					&
					& LV & SLV
					&
					& LV & SLV
					&
					& LV & SLV
					&
					& LV & SLV
					&
					& LV & SLV \\
					\midrule
					
					40
					&
					& \pnum{61.3137} & \pnum{61.9916}
					&
					& \pnum{64.2687} & \pnum{64.7066}
					&
					& \pnum{62.6676} & \pnum{63.1919}
					&
					& \pnum{67.4173} & \pnum{67.7442}
					&
					& \pnum{70.3474} & \pnum{70.3664} \\
					
					50
					&
					& \pnum{54.8427} & \pnum{55.6541}
					&
					& \pnum{57.7225} & \pnum{58.4461}
					&
					& \pnum{56.8391} & \pnum{57.5087}
					&
					& \pnum{61.0659} & \pnum{61.7249}
					&
					& \pnum{63.5809} & \pnum{64.1364} \\
					
					60
					&
					& \pnum{49.7008} & \pnum{50.5321}
					&
					& \pnum{52.5014} & \pnum{53.3586}
					&
					& \pnum{52.1215} & \pnum{52.7567}
					&
					& \pnum{55.9583} & \pnum{56.6622}
					&
					& \pnum{58.2352} & \pnum{58.9099} \\
					
					70
					&
					& \pnum{45.4982} & \pnum{46.4824}
					&
					& \pnum{48.2074} & \pnum{49.0731}
					&
					& \pnum{48.2055} & \pnum{48.8534}
					&
					& \pnum{51.7120} & \pnum{52.4821}
					&
					& \pnum{53.8019} & \pnum{54.5866} \\
					
					80
					&
					& \pnum{41.9924} & \pnum{42.9586}
					&
					& \pnum{44.6052} & \pnum{45.5538}
					&
					& \pnum{44.8931} & \pnum{45.5299}
					&
					& \pnum{48.1143} & \pnum{48.9069}
					&
					& \pnum{50.0429} & \pnum{50.8782} \\
					
					90
					&
					& \pnum{39.0189} & \pnum{39.9650}
					&
					& \pnum{41.5361} & \pnum{42.4429}
					&
					& \pnum{42.0489} & \pnum{42.6833}
					&
					& \pnum{45.0219} & \pnum{45.8353}
					&
					& \pnum{46.8088} & \pnum{47.6850} \\
					
					100
					&
					& \pnum{36.4610} & \pnum{37.4087}
					&
					& \pnum{38.8864} & \pnum{39.8802}
					&
					& \pnum{39.5760} & \pnum{40.2346}
					&
					& \pnum{42.3321} & \pnum{43.1876}
					&
					& \pnum{43.9937} & \pnum{44.9276} \\
					
					110
					&
					& \pnum{34.2613} & \pnum{35.2038}
					&
					& \pnum{36.5727} & \pnum{37.4799}
					&
					& \pnum{37.4031} & \pnum{38.0250}
					&
					& \pnum{39.9685} & \pnum{40.7947}
					&
					& \pnum{41.5189} & \pnum{42.4319} \\
					
					120
					&
					& \pnum{32.3313} & \pnum{33.1953}
					&
					& \pnum{34.5559} & \pnum{35.4911}
					&
					& \pnum{35.4763} & \pnum{36.1113}
					&
					& \pnum{37.8732} & \pnum{38.7202}
					&
					& \pnum{39.3244} & \pnum{40.2664} \\
					
					130
					&
					& \pnum{30.6236} & \pnum{31.4448}
					&
					& \pnum{32.7672} & \pnum{33.6649}
					&
					& \pnum{33.7538} & \pnum{34.3636}
					&
					& \pnum{36.0010} & \pnum{36.8245}
					&
					& \pnum{37.3634} & \pnum{38.2860} \\
					
					140
					&
					& \pnum{29.1016} & \pnum{29.9208}
					&
					& \pnum{31.1688} & \pnum{32.0708}
					&
					& \pnum{32.2123} & \pnum{32.8436}
					&
					& \pnum{34.3261} & \pnum{35.1751}
					&
					& \pnum{35.6089} & \pnum{36.5619} \\
					
					150
					&
					& \pnum{27.7362} & \pnum{28.5135}
					&
					& \pnum{29.7315} & \pnum{30.5951}
					&
					& \pnum{30.8264} & \pnum{31.4311}
					&
					& \pnum{32.8208} & \pnum{33.6423}
					&
					& \pnum{34.0322} & \pnum{34.9593} \\
					
					160
					&
					& \pnum{26.5040} & \pnum{27.2896}
					&
					& \pnum{28.4317} & \pnum{29.2238}
					&
					& \pnum{29.5664} & \pnum{30.1124}
					&
					& \pnum{31.4532} & \pnum{32.2114}
					&
					& \pnum{32.5999} & \pnum{33.4628} \\
					
					\bottomrule
				\end{tabular}%
			}
		}
		\caption{GMMB prices with discrete-time optimal surrender in the LV and SLV models.}
		\label{tab:surrender-lv-slv-results}
	\end{table}

	Table~\ref{tab:surrender-lv-slv-results} shows a systematic SLV premium
	relative to LV across all fee--surrender-adjustment combinations. The maximum
	absolute SLV--LV difference is about $0.994$ and is attained at $F_0=100$ for
	$(c,\kappa_s)=(1\%,1\%)$, while the maximum relative difference is about
	$2.96\%$ and is attained at $F_0=160$ for $(c,\kappa_s)=(0\%,0\%)$. These discrepancies are substantially larger than those observed in the benchmark without surrender. Together with the refinement study, they are consistent with the fact that the surrender feature depends on continuation values and on the conditional future dynamics of the fund, not only on the one-dimensional marginal distributions matched by the local-volatility calibration.
	
	The comparison between Tables~\ref{tab:european-lv-slv-results} and \ref{tab:surrender-lv-slv-results} is the central numerical message of the experiment. Without surrender, the LV and SLV values are nearly indistinguishable, as expected from a calibration based on the same projected volatility structure. When the surrender option is introduced, the SLV model produces a persistent difference, predominantly an SLV premium in the economically relevant cases. Stochastic volatility changes the continuation region and the timing value of the surrender right, even though the LV and SLV models share the same local-volatility target.

	Table~\ref{tab:fair-fee-lv-slv-results} reports the fair insurance fees
	$c_M^*$ for the GMMB with surrender option. Using the notation introduced
	in Section~\ref{sec:contract}, for each model
	$M\in\{\mathrm{LV},\mathrm{SLV}\}$ the numerical root is computed for
	\begin{equation}
		H_M(c,\kappa_s)=F_0e^{-cT}D_q(0,T)
		+U^{M}_{\mathrm{Sur}}(0,X_0;F_0,c,\kappa_s)-F_0.
	\end{equation}
	Here $D_q(0,T)=1$ in the synthetic experiment because $q_t\equiv0$. The scalar equation $H_M(c,\kappa_s)=0$ is solved by a Brent bracketing method. Each function evaluation recomputes the guarantee component with the corresponding backward solver, so that the fee is consistent with the same surrender rule used in the pricing tables. The root-finding residuals are at most about $8.9\cdot10^{-5}$ in contract-value units and are negligible relative to the fee differences reported below.
	
	\begin{table}[H]
		\centering
		
		\begin{tabular}{rrrrr}
			\toprule
			$\kappa_s$ & $c^*_{\mathrm{LV}}$ & $c^*_{\mathrm{SLV}}$ & $\Delta c^*$ (bp) & Relative increase \\
			\midrule
			$0\%$   & \pgfmathprintnumber[fixed,zerofill,precision=6]{0.055422} & \pgfmathprintnumber[fixed,zerofill,precision=6]{0.056216} & \num{7.94} & \num{1.43}\% \\
			$0.5\%$ & \pgfmathprintnumber[fixed,zerofill,precision=6]{0.056589} & \pgfmathprintnumber[fixed,zerofill,precision=6]{0.057485} & \num{8.96} & \num{1.58}\% \\
			$1\%$   & \pgfmathprintnumber[fixed,zerofill,precision=6]{0.057842} & \pgfmathprintnumber[fixed,zerofill,precision=6]{0.058830} & \num{9.88} & \num{1.71}\% \\
			\bottomrule
		\end{tabular}
		\normalsize
		\caption{Fair insurance fees for the GMMB with surrender option under LV and SLV. The last two columns report the SLV--LV absolute increment in annual basis points and the corresponding relative increase. The root is computed by a Brent method.}
		\label{tab:fair-fee-lv-slv-results}
	\end{table}
	
	For all three values of the surrender-adjustment parameter, the SLV fair
	insurance fee exceeds the corresponding LV fair insurance fee. The annual fee increments are
	$7.94$, $8.96$ and $9.88$ basis points for $\kappa_s=0$, $0.5\%$ and
	$1\%$, respectively, corresponding to relative increases of approximately
	$1.43\%$, $1.58\%$ and $1.71\%$. In both models the fair insurance fee
	increases with $\kappa_s$, consistently with
	Proposition~\ref{prop:fee-monotone}. For a fixed insurance fee, a larger
	$\kappa_s$ increases the immediate surrender payoff and therefore cannot
	decrease the value of the surrenderable guarantee. Under the monotonicity
	condition on the fair-fee equation stated in
	Proposition~\ref{prop:fee-monotone}, this translates into a nondecreasing
	fair insurance fee.

	We finally investigate the optimal surrender regions associated with the GMMB surrender option. In the local-volatility model, the stopping problem can be represented in the policy-account
	value $F$ only, since $F_t^{(c)}$ is in one-to-one correspondence with the reference fund level $S_t$ at each fixed time. The optimal surrender boundary can therefore be represented by a single curve
	\(B^{\mathrm{LV}}(t)\). By contrast, in the stochastic-local volatility model, the
	state variables are the policy-account value and $V$. The optimal surrender
	boundary is therefore a surface,
	\[
	F = B^{\mathrm{SLV}}(t,V),
	\]
	and the surrender region is
	\[
	\mathcal S^{\mathrm{SLV}}
	=
	\left\{(t,F,V): F \leq B^{\mathrm{SLV}}(t,V)\right\}.
	\]
	
	Figure~\ref{fig:surrender_regions} reports two representative cases. The left
	panel of each figure compares the LV boundary with several sections of the SLV
	surrender-boundary surface, obtained by fixing different levels of $V$.
	The right panel shows the corresponding two-dimensional surrender region in the
	$(F,V)$-plane at \(t=5\). The shaded region is the surrender region, while the
	white region is the continuation region.
	
	A clear pattern emerges. For fixed time, the SLV surrender boundary decreases as
	the level of \(V\) increases. This is consistent with the economic intuition
	that higher volatility increases the value of waiting, thereby reducing the incentive
	to surrender early. Consequently, the surrender region becomes smaller when the
	level of $V$ is high. The figures also show that increasing the surrender-adjustment
	parameter \(\kappa_s\) shifts the surrender boundary upward. This is expected from the
	immediate surrender payoff
	\[
	\left(G-e^{-\kappa_s(T-t)}F_t^{(c)}\right)^+,
	\]
	since a larger \(\kappa_s\) reduces the effective policy-account value entering the surrender
	payoff and therefore makes early surrender more attractive.
	
	These plots complement the pricing results reported above. Without surrender,
	LV and SLV prices are close because the payoff depends only on the
	distribution of the terminal policy-account value. With the surrender option, instead, the
	value depends on the continuation region and hence on the conditional future
	dynamics. The dependence of \(B^{\mathrm{SLV}}(t,V)\) on $V$ explains
	why prices with surrender under SLV differ systematically from those obtained under
	the pure LV specification.
	
	\begin{figure}[htbp]
		\centering
		
		\begin{subfigure}{\textwidth}
			\centering
			\includegraphics[width=\textwidth]{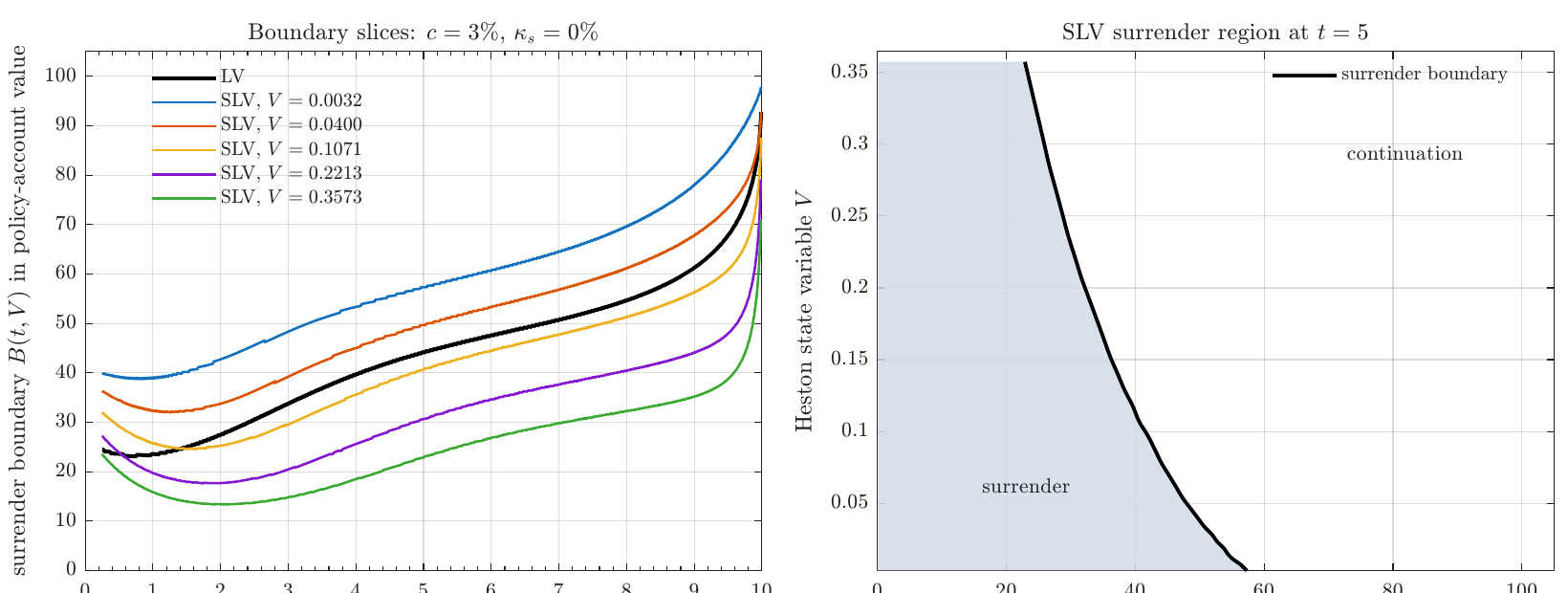}
			\caption{\(c=3\%\), \(\kappa_s=0\%\).}
			\label{fig:surrender_kappa0}
		\end{subfigure}
		
		\vspace{0.3cm}
		
		\begin{subfigure}{\textwidth}
			\centering
			\includegraphics[width=\textwidth]{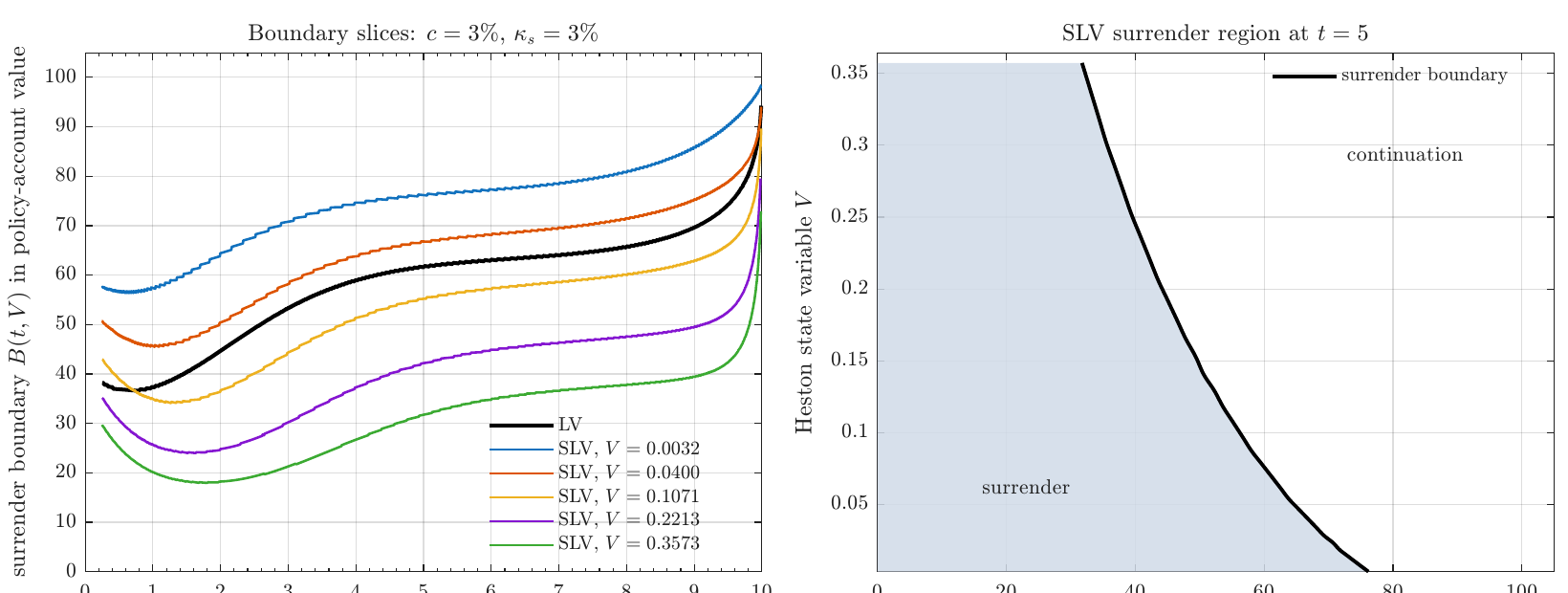}
			\caption{\(c=3\%\), \(\kappa_s=3\%\).}
			\label{fig:surrender_kappa3}
		\end{subfigure}
		
		\caption{Optimal surrender boundaries and surrender regions for the GMMB surrender option. In each row, the left panel reports the LV boundary and
			selected SLV boundary slices for different levels of $V$. The right panel
			reports the SLV surrender region in the $(F,V)$-plane at \(t=5\).}
		\label{fig:surrender_regions}
	\end{figure}
	
	\FloatBarrier
	
	\subsection{Market Data}
	\label{subsec:real-market-eurostoxx}
	
	The synthetic experiment isolates the effect of stochastic volatility under a controlled non-flat volatility structure. We now repeat the same LV--SLV comparison using EURO STOXX 50 market data, including direct market-data counterparts of the pricing and surrender-region analyses reported above. The purpose is twofold: first, to construct the common LV marginal target and the deterministic curves from the observed vanilla-option information, while informing the Heston parameters from additional volatility-market data; second, to assess whether differences between LV and SLV remain economically relevant under an option-calibrated local-volatility surface.
	
	The market construction is deliberately sequential. We first extract deterministic discount and dividend-yield term structures from matched call--put quotes through put--call parity. We then calibrate the LV surface to the European option panel. The Heston parameters are selected separately using volatility-market and joint equity--volatility diagnostics, after which the SLV leverage function is recalibrated to the same LV target. Thus the deterministic curves are fixed before, and independently of, the specification of $V_t$.
	
	The valuation date is 15 July 2026 and the reference index level is
	$S_0=6265.58$. The option snapshot, obtained from Bloomberg, contains matched
	EURO STOXX 50 call and put quotes. The matched call--put pairs are first used
	to infer the deterministic discount and dividend-yield term structures through
	put--call parity. After this parity-based preprocessing and the exclusion of
	inactive observations, maturities shorter than $0.01$ years and option prices
	below $0.001S_0$, the local-volatility calibration itself is carried out using
	$2070$ European put prices. This filtered put sample spans $27$ maturities from
	$0.1014$ to $9.4411$ years and is partitioned into short-, medium- and
	long-dated groups to preserve adequate temporal resolution throughout the
	calibration.
	
	At each quoted maturity $T_i$, let $C(K,T_i)$ and $P(K,T_i)$ denote the
	corresponding call and put prices at strike $K$. Put--call parity gives
	\begin{equation}
		C(K,T_i)-P(K,T_i)=S_0D_q(0,T_i)-K D_r(0,T_i),
		\label{eq:market-put-call-parity}
	\end{equation}
	so that the affine dependence on strike has slope $-D_r(0,T_i)$ and intercept
	$S_0D_q(0,T_i)$. Since $S_0$ is known, these two coefficients identify the
	discount factor $D_r(0,T_i)$ and the dividend factor $D_q(0,T_i)$. We define
	\begin{equation}
		H_r(T_i)=-\log D_r(0,T_i),
		\qquad
		H_q(T_i)=-\log D_q(0,T_i),
		\label{eq:market-cumulative-curves}
	\end{equation}
	and interpolate the cumulative curves linearly between the maturity nodes used
	in the numerical implementation. This yields piecewise-constant instantaneous
	rates. Each numerical step uses the exact increments of $H_r$ and $H_q$,
	including when it crosses a curve knot.
	Figure~\ref{fig:real-market-rq-curves} reports the instantaneous term
	structures used in the calculations.
	
	\begin{figure}[htbp]
		\centering
		\includegraphics[width=0.8\textwidth]{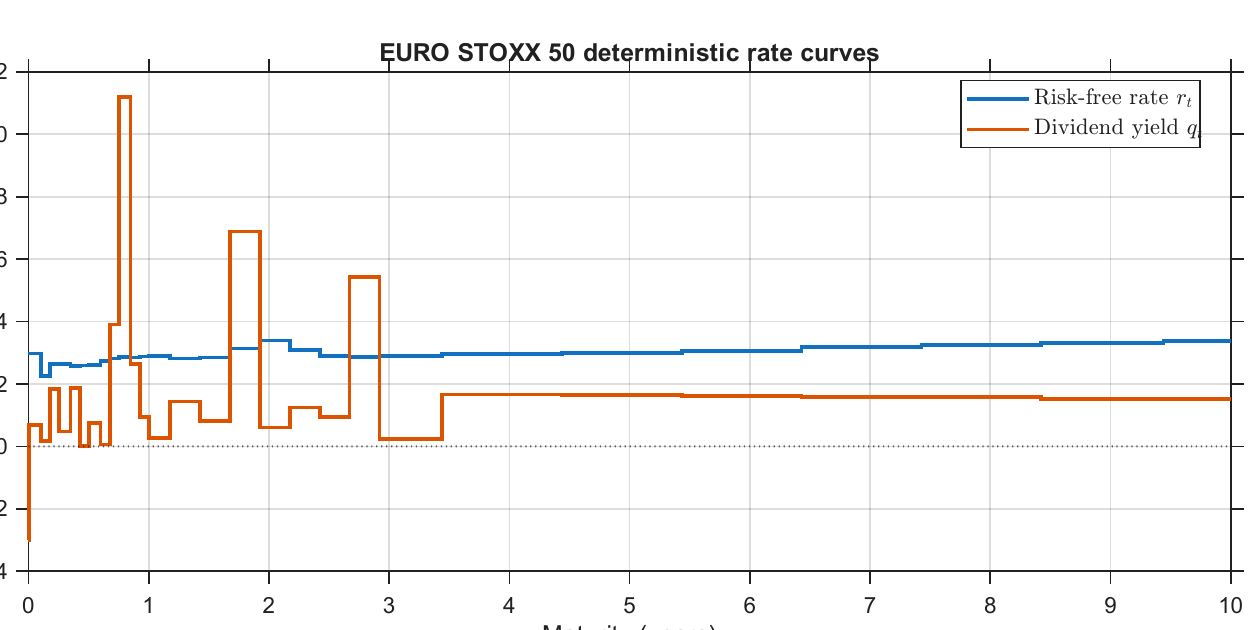}
		\caption{EURO STOXX 50 piecewise-constant instantaneous term structures used in the pricing calculations: discount-rate curve $r_t$ and dividend yield $q_t$. The cumulative curves $H_r$ and $H_q$ are constructed from put--call parity and linearly interpolated between quoted maturity nodes.}
		\label{fig:real-market-rq-curves}
	\end{figure}
	
	The local-volatility input is constructed by applying the neural calibration procedure for an arbitrage-free and complete recombining binomial market model developed in \citet{MolentVellekoop2026}. The calibrated discrete model is converted into the dense local-volatility representation used by the LV and SLV pricing codes and is evaluated by bilinear interpolation in time and log-moneyness, with boundary-cell extrapolation outside the tabulated cell. The calibration criterion is maturity-balanced so that short- and long-dated observations contribute comparably to the objective. Pricing the calibration instruments with the trinomial scheme used during calibration gives a maturity-balanced relative root-mean-square error (RMSE) of $3.1\%$. Repricing the same option set with an independent Crank--Nicolson solver gives $3.3\%$, while the interpolated representation used in the subsequent SLV calculations gives a virtually identical error of $3.3\%$. These are in-sample calibration and numerical-validation measures rather than out-of-sample forecasting statistics. The self-consistent learning construction of \citet{Wang2025} remains the source of the synthetic local-volatility surface used in Section~\ref{subsec:synthetic-data}, but it is not the calibration method used for the EURO STOXX 50 surface.
	
	\begin{rem}
		The approximately $3\%$ errors reported above measure the fit of the
		calibrated local-volatility surface to the observed EURO STOXX 50 option
		panel; they are not LV--SLV discrepancies. The SLV leverage function is
		subsequently calibrated to reproduce the same local-volatility projection,
		so both models share the same option-calibrated LV target. Consequently,
		residual LV--SLV differences for European payoffs are interpreted as
		numerical calibration and discretisation errors. By contrast, differences
		for surrenderable contracts need not vanish, since their values depend on
		conditional continuation dynamics that are not determined by the common
		one-date marginals. These latter differences constitute the model effect
		investigated in this paper. We do not use direct repricing of the full
		$2070$-option panel under SLV as an additional calibration objective;
		accordingly, the $3\%$ figure should not be interpreted as an independently
		estimated SLV-to-market error.
	\end{rem}
	
	The deterministic discount and dividend-yield curves described above are already fixed before any Heston parameter is selected and are used unchanged in both the LV and SLV calculations. The Heston parameters entering the SLV model cannot be identified uniquely from the same vanilla surface used to determine the local volatility. For notational convenience in the parameter-selection step, write $v_0:=V_0$. For a generic parameter vector
	\[
	\eta=(v_0,\kappa_V,\theta_V,\omega,\rho),
	\]
	the leverage function is recalibrated through \eqref{eq:leverage-condition}. Changes in the dynamics of $V_t$ can therefore be partially offset by changes in $L(t,S)$ while preserving the one-dimensional spot marginals; the exact scale indeterminacy is stated in Remark~\ref{rem:variance-scale}. For this reason, a stand-alone Heston calibration to the same EURO STOXX 50 vanilla panel would not, by itself, identify a unique structural parameterization of the SLV decomposition. We instead use additional information from the volatility market and from the joint equity--volatility dynamics, and subsequently recalibrate the leverage function to the same local-volatility target.
	
	The parameters $(v_0,\kappa_V,\theta_V)$ are informed by the contemporaneous VSTOXX futures term structure reported in the Eurex daily statistics for the valuation date. For the actively traded maturities from July to November 2026, the observed levels are
	\[
	17.45,\qquad 18.40,\qquad 19.35,\qquad 20.00,\qquad 20.30,
	\]
	at ACT/365 maturities
	\[
	0.01918,\qquad 0.09589,\qquad 0.17260,\qquad 0.26849,\qquad 0.34521.
	\]
	For the sole purpose of obtaining a parsimonious diagnostic, we exploit the CIR dynamics of $V_t$. Conditional on the current state $V_0=v_0$,
	\[
	\E[V_u\mid V_0=v_0]=\theta_V+(v_0-\theta_V)e^{-\kappa_V u}.
	\]
	Set $\tau=30/365$. For a VSTOXX futures maturity $T_f$, the expected average of $V_t$ over the forward 30-day window $[T_f,T_f+\tau]$ is
	\begin{equation}
		\overline V(T_f)
		=\frac{1}{\tau}\int_{T_f}^{T_f+\tau}\E[V_u\mid V_0=v_0] \,\dd u
		=\theta_V+(v_0-\theta_V)e^{-\kappa_V T_f}
		\frac{1-e^{-\kappa_V\tau}}{\kappa_V\tau}.
		\label{eq:vstoxx-heston-variance-proxy}
	\end{equation}
	Let $F^{\mathrm{VSTOXX}}(T_f)$ denote the observed VSTOXX futures level with maturity $T_f$. We select $(v_0,\kappa_V,\theta_V)$ by nonlinear least squares against the diagnostic proxy
	\begin{equation}
		F^{\mathrm{VSTOXX}}(T_f)\approx100\sqrt{\overline V(T_f)}.
		\label{eq:vstoxx-futures-proxy}
	\end{equation}
	The fit gives
	\[
	v_0=0.0264667,\qquad
	\kappa_V=3.7092,\qquad
	\theta_V=0.0460807,
	\]
	with an RMSE of approximately $0.0635$ VSTOXX index points. The values used in the SLV calculations are rounded to
	\[
	v_0=0.0265,\qquad
	\kappa_V=3.70,\qquad
	\theta_V=0.0461,
	\]
	corresponding to $\sqrt{v_0}=16.28\%$ and $\sqrt{\theta_V}=21.47\%$.
	
	\begin{table}[H]
		\centering
		\caption{VSTOXX futures term structure used to inform the parameters $(v_0,\kappa_V,\theta_V)$ governing $V_t$, together with the corresponding Heston-based VSTOXX proxy.}
		\label{tab:vstoxx-futures-fit}
		\begin{tabular}{rrr}
			\toprule
			Maturity (years) & Market level & Heston-based VSTOXX proxy \\
			\midrule
			0.01918 & 17.45 & 17.4170 \\
			0.09589 & 18.40 & 18.5026 \\
			0.17260 & 19.35 & 19.2790 \\
			0.26849 & 20.00 & 19.9589 \\
			0.34521 & 20.30 & 20.3426 \\
			\bottomrule
		\end{tabular}
	\end{table}
	
	This proxy fit is used only to identify a plausible region for the Heston parameters governing $V_t$; it is not a pricing model for VSTOXX futures. VSTOXX is an option-implied volatility index constructed from EURO STOXX 50 options and represents the square root of implied variance over a fixed tenor, while a VSTOXX future references the future level of that index. Equation~\eqref{eq:vstoxx-futures-proxy} instead replaces this object by the square root of the expected average of $V_t$. The distinction is even more important in the SLV model, where the instantaneous equity variance rate is $L^2(t,S_t)V_t$ rather than $V_t$ alone. The reported RMSE should therefore be read only as the goodness of fit of this low-dimensional diagnostic proxy. In addition, the available liquid futures cover only a short segment of the term structure, so $\kappa_V$ and especially $\theta_V$ are not sharply identified separately.
	
	The remaining parameters are informed by the historical joint dynamics of the
	EURO STOXX 50 and VSTOXX indices. Daily closing levels of the two indices are
	used to construct diagnostics of equity--volatility dependence and
	of the diffusion scale of $V_t$ over windows ending on the valuation date. The
	correlation between their daily log changes is close to $-0.80$ over one-,
	three- and five-year windows. Because VSTOXX is not the Heston volatility component
	$\sqrt{V_t}$, this empirical correlation is not interpreted as a direct
	estimate of $\rho$; rather, it provides a structural anchor for the pronounced
	negative equity--volatility dependence. We therefore set $\rho=-0.75$ as a
	nearby structural value. Similarly, diagnostics based on the dynamics of
	squared VSTOXX levels indicate values of the diffusion coefficient of $V_t$ of roughly
	$0.5$--$0.6$, motivating the choice $\omega=0.50$. These quantities are used
	only as structural anchors and not as risk-neutral parameter estimates.
	
	Combining these inputs gives
	\begin{equation}
		v_0=0.0265,\qquad
		\kappa_V=3.70,\qquad
		\theta_V=0.0461,\qquad
		\omega=0.50,\qquad
		\rho=-0.75.
		\label{eq:market-heston-parameters}
	\end{equation}
	The Feller condition is satisfied,
	\[
	2\kappa_V\theta_V=0.34114>\omega^2=0.25.
	\]
	The parameter vector in \eqref{eq:market-heston-parameters} should therefore be interpreted as a market-informed structural specification rather than as a uniquely identified risk-neutral Heston calibration. The option-calibrated local-volatility surface determines the target one-date spot marginals through the projection condition, while the volatility-market data provide additional information on the dynamics of $V_t$.

	\begin{rem}
		\label{rem:variance-scale}
		The decomposition of the instantaneous equity variance rate into
		$L^2(t,S_t)V_t$ is subject to an exact scale indeterminacy. For any
		$a>0$, define
		\[
		\widetilde V_t=aV_t,
		\qquad
		\widetilde V_0=aV_0,
		\qquad
		\widetilde\theta_V=a\theta_V,
		\qquad
		\widetilde\omega=\sqrt{a}\,\omega,
		\qquad
		\widetilde L(t,s)=\frac{L(t,s)}{\sqrt a},
		\]
		while leaving $\kappa_V$ and $\rho$ unchanged. Then
		$\widetilde V$ has the same CIR form as $V$ and
		\[
		\widetilde L(t,S_t)\sqrt{\widetilde V_t}
		=L(t,S_t)\sqrt{V_t},
		\qquad
		\widetilde L^2(t,S_t)\widetilde V_t
		=L^2(t,S_t)V_t.
		\]
		Hence the spot dynamics and the Markovian-projection identity do not
		identify the scale of $V_t$ by themselves. In the present
		market application, this normalisation is fixed by the external
		volatility-market information used to anchor the scale of $V_t$,
		rather than chosen arbitrarily; concretely, it is anchored by the
		VSTOXX-informed specification in
		\eqref{eq:market-heston-parameters}.
	\end{rem}

	Once the stochastic-volatility parameters have been fixed, the leverage function is recalibrated to the same EURO STOXX 50 local-volatility surface used by the LV benchmark. The results reported below use $4000$ time steps and $4000$ spatial nodes over the ten-year horizon, consistently with the fine-grid specification adopted for the market-data experiment. In the market calibration, four leverage fixed-point passes are applied at each time level: the first uses the direct projection update and the next three use relaxation with weight $\alpha=0.5$. The conditional mean of $V_t$ entering the leverage update is floored with $\varepsilon_V=10^{-12}$. The transformed-coordinate domain uses the tail tolerance $\varepsilon_{\rm fd}=10^{-7}$ in the rule given in Appendix~\ref{app:numerical-details}, and the time step may be split into at most eight forward and four backward implicit substeps. The forward density is propagated with the conservative fully implicit finite-volume/upwind scheme described in Appendix~\ref{app:numerical-details}, with zero-total-flux boundaries. The implementation includes numerical safeguards for low-density regions when
	forming the conditional mean of $V_t$. Since the longest retained option maturity is $9.4411$ years, the calibrated
	local-volatility surface is linearly extrapolated over the remaining part of
	the ten-year horizon from its final calibrated time cell; the ten-year results
	should therefore be interpreted with this qualification.
	
	\begin{figure}[H]
		\centering
		\includegraphics[width=\textwidth]{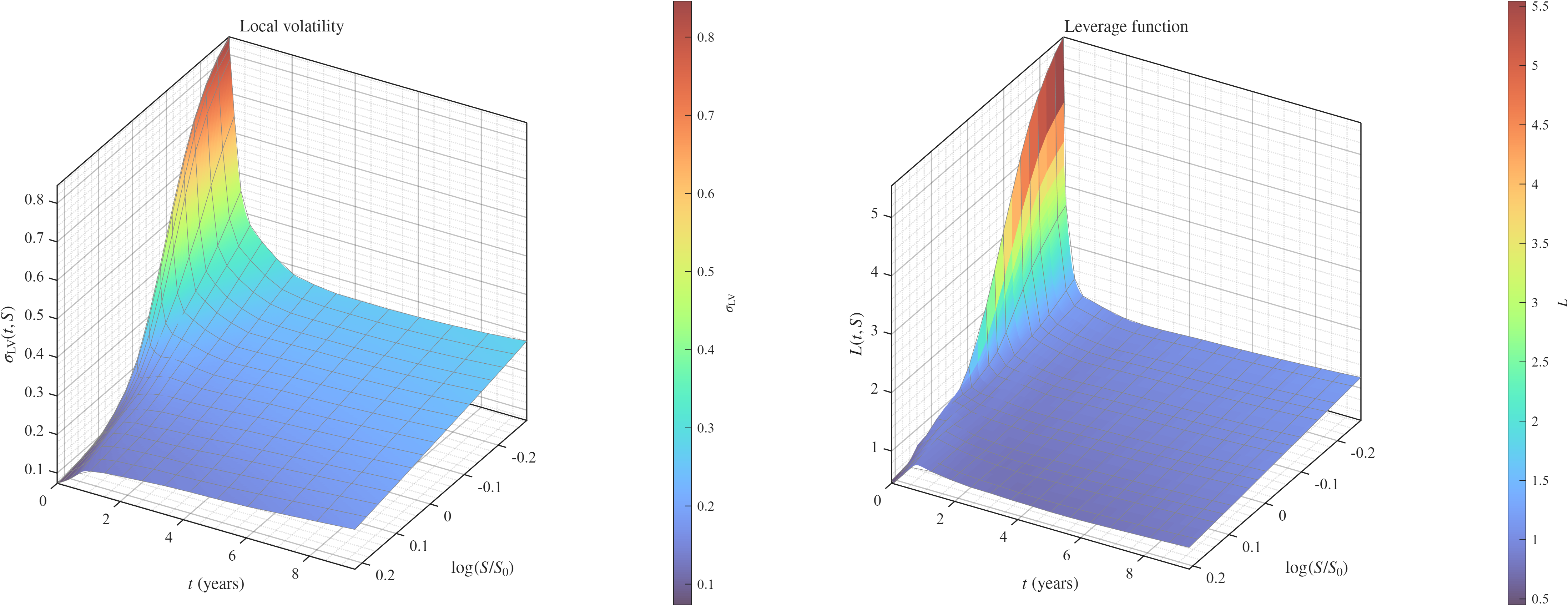}
		\caption{EURO STOXX 50 local-volatility surface (left) and corresponding SLV leverage function (right), shown against maturity and log-moneyness $\log(S/S_0)$.}
		\label{fig:real-market-lv-leverage}
	\end{figure}

	We next repeat the GMMB experiment without surrender under the option-calibrated local-volatility surface and the corresponding market-informed Heston SLV specification. The contract parameters remain $G=100$, $T=10$ years and $F_0\in\{40,50,\ldots,160\}$. For a fixed fee $c$, define
	\[
	\beta=\frac{F_0e^{-cT}}{S_0},
	\qquad
	K_{\mathrm{eq}}=\frac{GS_0e^{cT}}{F_0}.
	\]
	The terminal guarantee can then be written as the scaled European put
	\begin{equation}
		\left(G-F_0\frac{S_T}{S_0}e^{-cT}\right)^+
		=\beta\left(K_{\mathrm{eq}}-S_T\right)^+.
		\label{eq:real-market-gmmb-put-equivalence}
	\end{equation}
	As in the synthetic experiment without surrender, $\kappa_s$ does not enter the terminal payoff, so the three columns with $c=3\%$ coincide by construction.
	
	\begin{table}[H]
		\centering
		\renewcommand{\arraystretch}{1.15}
		
		\resizebox{\textwidth}{!}{%
			\begin{tabular}{rccccccccccccccc}
				\toprule
				&
				&
				\multicolumn{2}{c}{$c=0\%,\,\kappa_s=0\%$}
				&
				&
				\multicolumn{2}{c}{$c=1\%,\,\kappa_s=1\%$}
				&
				&
				\multicolumn{2}{c}{$c=3\%,\,\kappa_s=0\%$}
				&
				&
				\multicolumn{2}{c}{$c=3\%,\,\kappa_s=2\%$}
				&
				&
				\multicolumn{2}{c}{$c=3\%,\,\kappa_s=3\%$} \\
				
				\cmidrule(lr){3-4}
				\cmidrule(lr){6-7}
				\cmidrule(lr){9-10}
				\cmidrule(lr){12-13}
				\cmidrule(lr){15-16}
				
				$F_0$
				&
				& LV & SLV
				&
				& LV & SLV
				&
				& LV & SLV
				&
				& LV & SLV
				&
				& LV & SLV \\
				\midrule
				
				40
				&
				& \pnum{39.9587} & \pnum{39.9533}
				&
				& \pnum{42.9504} & \pnum{42.9458}
				&
				& \pnum{48.3317} & \pnum{48.3284}
				&
				& \pnum{48.3317} & \pnum{48.3284}
				&
				& \pnum{48.3317} & \pnum{48.3284} \\
				
				50
				&
				& \pnum{32.8089} & \pnum{32.8028}
				&
				& \pnum{36.0673} & \pnum{36.0612}
				&
				& \pnum{42.2740} & \pnum{42.2693}
				&
				& \pnum{42.2740} & \pnum{42.2693}
				&
				& \pnum{42.2740} & \pnum{42.2693} \\
				
				60
				&
				& \pnum{26.9366} & \pnum{26.9321}
				&
				& \pnum{30.1223} & \pnum{30.1166}
				&
				& \pnum{36.6366} & \pnum{36.6306}
				&
				& \pnum{36.6366} & \pnum{36.6306}
				&
				& \pnum{36.6366} & \pnum{36.6306} \\
				
				70
				&
				& \pnum{22.4227} & \pnum{22.4217}
				&
				& \pnum{25.2812} & \pnum{25.2777}
				&
				& \pnum{31.6155} & \pnum{31.6095}
				&
				& \pnum{31.6155} & \pnum{31.6095}
				&
				& \pnum{31.6155} & \pnum{31.6095} \\
				
				80
				&
				& \pnum{19.1375} & \pnum{19.1405}
				&
				& \pnum{21.5371} & \pnum{21.5370}
				&
				& \pnum{27.3213} & \pnum{27.3166}
				&
				& \pnum{27.3213} & \pnum{27.3166}
				&
				& \pnum{27.3213} & \pnum{27.3166} \\
				
				90
				&
				& \pnum{16.7496} & \pnum{16.7561}
				&
				& \pnum{18.7475} & \pnum{18.7510}
				&
				& \pnum{23.7774} & \pnum{23.7751}
				&
				& \pnum{23.7774} & \pnum{23.7751}
				&
				& \pnum{23.7774} & \pnum{23.7751} \\
				
				100
				&
				& \pnum{14.9833} & \pnum{14.9925}
				&
				& \pnum{16.6516} & \pnum{16.6582}
				&
				& \pnum{20.9499} & \pnum{20.9505}
				&
				& \pnum{20.9499} & \pnum{20.9505}
				&
				& \pnum{20.9499} & \pnum{20.9505} \\
				
				110
				&
				& \pnum{13.6615} & \pnum{13.6730}
				&
				& \pnum{15.0549} & \pnum{15.0641}
				&
				& \pnum{18.7330} & \pnum{18.7365}
				&
				& \pnum{18.7330} & \pnum{18.7365}
				&
				& \pnum{18.7330} & \pnum{18.7365} \\
				
				120
				&
				& \pnum{12.6706} & \pnum{12.6838}
				&
				& \pnum{13.8269} & \pnum{13.8380}
				&
				& \pnum{16.9783} & \pnum{16.9844}
				&
				& \pnum{16.9783} & \pnum{16.9844}
				&
				& \pnum{16.9783} & \pnum{16.9844} \\
				
				130
				&
				& \pnum{11.9339} & \pnum{11.9486}
				&
				& \pnum{12.8801} & \pnum{12.8929}
				&
				& \pnum{15.5762} & \pnum{15.5845}
				&
				& \pnum{15.5762} & \pnum{15.5845}
				&
				& \pnum{15.5762} & \pnum{15.5845} \\
				
				140
				&
				& \pnum{11.3946} & \pnum{11.4108}
				&
				& \pnum{12.1544} & \pnum{12.1686}
				&
				& \pnum{14.4475} & \pnum{14.4577}
				&
				& \pnum{14.4475} & \pnum{14.4577}
				&
				& \pnum{14.4475} & \pnum{14.4577} \\
				
				150
				&
				& \pnum{11.0125} & \pnum{11.0297}
				&
				& \pnum{11.6038} & \pnum{11.6194}
				&
				& \pnum{13.5354} & \pnum{13.5471}
				&
				& \pnum{13.5354} & \pnum{13.5471}
				&
				& \pnum{13.5354} & \pnum{13.5471} \\
				
				160
				&
				& \pnum{10.7319} & \pnum{10.7497}
				&
				& \pnum{11.1947} & \pnum{11.2116}
				&
				& \pnum{12.7987} & \pnum{12.8116}
				&
				& \pnum{12.7987} & \pnum{12.8116}
				&
				& \pnum{12.7987} & \pnum{12.8116} \\
				
				\bottomrule
			\end{tabular}%
		}
		
		\caption{European GMMB guarantee prices under the EURO STOXX 50 option-calibrated local-volatility surface and the market-informed Heston SLV specification. The deterministic no-arbitrage curves are used in both models.}
		\label{tab:real-market-european-gmmb}
	\end{table}
	
	The comparison without surrender is very tight. Since the surrender-adjustment
	parameter does not enter the terminal payoff, the three displayed columns with
	$c=3\%$ represent the same terminal-only contract. The table therefore contains
	$39$ distinct $(F_0,c)$ cases, corresponding to $13$ initial account values and
	three fee levels. Across these cases, the RMS relative SLV--LV discrepancy is
	approximately $0.0746\%$ and the mean absolute relative discrepancy is
	approximately $0.0553\%$. The maximum absolute difference is $0.0178$ and the
	maximum relative difference is approximately $0.166\%$. If all $65$ displayed
	entries are counted, including the repeated $c=3\%$ columns, the RMS relative
	discrepancy is $0.0649\%$. The residual differences change sign across the table
	and are consistent with residual numerical discretisation error rather than
	with an economically meaningful stochastic-volatility premium.
	
	We finally introduce discrete-time optimal surrender while keeping the same local-volatility target, deterministic term structures and stochastic-volatility parameters. The same surrender convention used in the synthetic experiment is retained, so that differences between the two models can be interpreted in terms of continuation dynamics rather than contractual changes.

	\begin{table}[H]
		\centering
		\renewcommand{\arraystretch}{1.15}
		
		\resizebox{\textwidth}{!}{%
			\begin{tabular}{rccccccccccccccc}
				\toprule
				&
				&
				\multicolumn{2}{c}{$c=0\%,\,\kappa_s=0\%$}
				&
				&
				\multicolumn{2}{c}{$c=1\%,\,\kappa_s=1\%$}
				&
				&
				\multicolumn{2}{c}{$c=3\%,\,\kappa_s=0\%$}
				&
				&
				\multicolumn{2}{c}{$c=3\%,\,\kappa_s=2\%$}
				&
				&
				\multicolumn{2}{c}{$c=3\%,\,\kappa_s=3\%$} \\
				
				\cmidrule(lr){3-4}
				\cmidrule(lr){6-7}
				\cmidrule(lr){9-10}
				\cmidrule(lr){12-13}
				\cmidrule(lr){15-16}
				
				$F_0$
				&
				& LV & SLV
				&
				& LV & SLV
				&
				& LV & SLV
				&
				& LV & SLV
				&
				& LV & SLV \\
				\midrule
				
				40
				&
				& \pnum{59.9805} & \pnum{59.9766}
				&
				& \pnum{63.7882} & \pnum{63.7843}
				&
				& \pnum{59.9835} & \pnum{59.9809}
				&
				& \pnum{67.2343} & \pnum{67.2309}
				&
				& \pnum{70.3509} & \pnum{70.3474} \\
				
				50
				&
				& \pnum{49.9775} & \pnum{49.9732}
				&
				& \pnum{54.7371} & \pnum{54.7329}
				&
				& \pnum{50.1326} & \pnum{50.2001}
				&
				& \pnum{59.0447} & \pnum{59.0411}
				&
				& \pnum{62.9405} & \pnum{62.9366} \\
				
				60
				&
				& \pnum{39.9745} & \pnum{39.9700}
				&
				& \pnum{45.6859} & \pnum{45.6816}
				&
				& \pnum{42.0472} & \pnum{42.2010}
				&
				& \pnum{50.8551} & \pnum{50.8515}
				&
				& \pnum{55.5301} & \pnum{55.5260} \\
				
				70
				&
				& \pnum{30.1687} & \pnum{30.6521}
				&
				& \pnum{36.6348} & \pnum{36.6303}
				&
				& \pnum{35.5730} & \pnum{35.7489}
				&
				& \pnum{42.6656} & \pnum{42.6619}
				&
				& \pnum{48.1197} & \pnum{48.1154} \\
				
				80
				&
				& \pnum{24.0911} & \pnum{24.7378}
				&
				& \pnum{28.4190} & \pnum{28.9920}
				&
				& \pnum{30.3829} & \pnum{30.5636}
				&
				& \pnum{34.8300} & \pnum{35.1830}
				&
				& \pnum{40.7092} & \pnum{40.7048} \\
				
				90
				&
				& \pnum{20.2215} & \pnum{20.8346}
				&
				& \pnum{23.4333} & \pnum{24.0792}
				&
				& \pnum{26.2485} & \pnum{26.4314}
				&
				& \pnum{29.1863} & \pnum{29.6682}
				&
				& \pnum{33.2988} & \pnum{33.3448} \\
				
				100
				&
				& \pnum{17.5498} & \pnum{18.1118}
				&
				& \pnum{20.0697} & \pnum{20.6804}
				&
				& \pnum{23.0084} & \pnum{23.1950}
				&
				& \pnum{25.0281} & \pnum{25.5200}
				&
				& \pnum{27.3115} & \pnum{27.9188} \\
				
				110
				&
				& \pnum{15.6193} & \pnum{16.1296}
				&
				& \pnum{17.6559} & \pnum{18.2203}
				&
				& \pnum{20.4957} & \pnum{20.6818}
				&
				& \pnum{21.9226} & \pnum{22.3945}
				&
				& \pnum{23.4090} & \pnum{24.0549} \\
				
				120
				&
				& \pnum{14.2852} & \pnum{14.6697}
				&
				& \pnum{15.8562} & \pnum{16.3761}
				&
				& \pnum{18.5393} & \pnum{18.7102}
				&
				& \pnum{19.5487} & \pnum{19.9950}
				&
				& \pnum{20.5774} & \pnum{21.1960} \\
				
				130
				&
				& \pnum{13.3631} & \pnum{13.6511}
				&
				& \pnum{14.5552} & \pnum{14.9738}
				&
				& \pnum{17.0031} & \pnum{17.1559}
				&
				& \pnum{17.6930} & \pnum{18.1160}
				&
				& \pnum{18.4329} & \pnum{19.0137} \\
				
				140
				&
				& \pnum{12.6977} & \pnum{12.9513}
				&
				& \pnum{13.6361} & \pnum{13.9463}
				&
				& \pnum{15.7779} & \pnum{15.9230}
				&
				& \pnum{16.2326} & \pnum{16.6195}
				&
				& \pnum{16.7605} & \pnum{17.3047} \\
				
				150
				&
				& \pnum{12.2193} & \pnum{12.4575}
				&
				& \pnum{12.9560} & \pnum{13.2205}
				&
				& \pnum{14.7898} & \pnum{14.9346}
				&
				& \pnum{15.1128} & \pnum{15.4259}
				&
				& \pnum{15.4409} & \pnum{15.9419} \\
				
				160
				&
				& \pnum{11.8632} & \pnum{12.0921}
				&
				& \pnum{12.4491} & \pnum{12.6942}
				&
				& \pnum{13.9893} & \pnum{14.1366}
				&
				& \pnum{14.2399} & \pnum{14.4952}
				&
				& \pnum{14.4497} & \pnum{14.8551} \\
				
				\bottomrule
			\end{tabular}%
		}
		
		\caption{GMMB prices with discrete-time optimal surrender under the EURO STOXX 50 option-calibrated local-volatility surface and the market-informed Heston SLV specification. The deterministic no-arbitrage curves are used in both models.}
		\label{tab:real-market-surrender-gmmb}
	\end{table}

	The comparison with discrete-time optimal surrender is materially different from the benchmark without surrender. Across the $65$ entries in Table~\ref{tab:real-market-surrender-gmmb}, the absolute RMSE of the SLV--LV difference is approximately $0.3543$, the mean absolute relative difference is $1.464\%$, and the RMS relative difference is $1.884\%$. The maximum absolute difference is $0.6467$, attained at $F_0=80$ for $(c,\kappa_s)=(0,0)$, while the maximum relative difference is approximately $3.279\%$, attained at $F_0=120$ for $(c,\kappa_s)=(1\%,1\%)$. A small number of low-$F_0$ cases display negative SLV--LV differences of only a few thousandths of a contract-value unit; these are negligible relative to the economically relevant discrepancies observed once early surrender becomes material.
	
	The contrast between terminal-only and surrenderable contracts is the main
	result of the market-data experiment. Without surrender, all reported
	LV--SLV relative discrepancies remain below approximately $0.17\%$ and are
	consistent with residual numerical error. Once surrender is introduced, the
	RMS relative difference rises to $1.884\%$ and the maximum relative difference
	reaches approximately $3.279\%$. These differences are an order of magnitude
	larger than the numerical discrepancies observed for the corresponding terminal-only contracts on the same fine grid.
	They are therefore interpreted as economically meaningful model differences
	arising from conditional continuation dynamics and the state-dependent
	surrender decision. The results provide evidence that matching the one-date
	marginals implied by vanilla options does not eliminate model risk for
	surrenderable insurance liabilities, subject to the parameter-identification
	and extrapolation qualifications discussed above.
	
	To complete the market-data comparison, we compute the fair insurance fees
	$c_M^*$ for the same GMMB with surrender option. We set $F_0=G=100$ and
	$T=10$ years and consider $\kappa_s\in\{0,0.005,0.010\}$, as in the
	synthetic fair-fee experiment. Using the notation introduced in Section~\ref{sec:contract}, the fair-fee equation for each model
	$M\in\{\mathrm{LV},\mathrm{SLV}\}$ becomes
	\begin{equation}
		H_M(c,\kappa_s)=F_0e^{-cT}D_q(0,T)
		+U^M_{\mathrm{Sur}}(0,X_0;F_0,c,\kappa_s)-F_0=0,
		\label{eq:fair-fee-market}
	\end{equation}
	where $D_q(0,T)=\exp\{-\int_0^T q_u\,du\}$ is the dividend discount
	factor inferred from the option market. The leverage function is calibrated once to the option-calibrated local-volatility surface and then kept fixed during the root search, since the insurance fee is a contractual parameter rather than a parameter of the market dynamics. Each evaluation of $H_M(c,\kappa_s)$ recomputes the guarantee component with the corresponding backward solver, and the scalar root is obtained with the same Brent procedure used in the synthetic experiment.
	
	\begin{table}[H]
		\centering
		\begin{tabular}{rrrrr}
			\toprule
			$\kappa_s$ & $c^*_{\mathrm{LV}}$ & $c^*_{\mathrm{SLV}}$ & $\Delta c^*$ (bp) & Relative increase \\
			\midrule
			$0\%$   & \pgfmathprintnumber[fixed,zerofill,precision=6]{0.002787} & \pgfmathprintnumber[fixed,zerofill,precision=6]{0.003569} & \num{7.82} & \num{28.06}\% \\
			$0.5\%$ & \pgfmathprintnumber[fixed,zerofill,precision=6]{0.003446} & \pgfmathprintnumber[fixed,zerofill,precision=6]{0.004347} & \num{9.01} & \num{26.15}\% \\
			$1\%$   & \pgfmathprintnumber[fixed,zerofill,precision=6]{0.004296} & \pgfmathprintnumber[fixed,zerofill,precision=6]{0.005320} & \num{10.24} & \num{23.84}\% \\
			\bottomrule
		\end{tabular}
		\caption{Fair insurance fees for the GMMB with surrender option under the EURO STOXX 50 market-informed specification. The last two columns report the SLV--LV absolute increment in annual basis points and the corresponding relative increase. The root is computed by a Brent method with $F_0=G=100$ and $T=10$ years.}
		\label{tab:market-fair-fees}
	\end{table}
	
	Table~\ref{tab:market-fair-fees} shows that the SLV fair insurance fee is
	systematically higher than the LV fair insurance fee. The annual-fee differences are
	approximately $7.82$, $9.01$ and $10.24$ basis points for $\kappa_s=0$,
	$0.5\%$ and $1\%$, respectively. In absolute terms these increments are
	close to the $7.94$--$9.88$ basis-point range obtained in the synthetic
	experiment. In relative terms, however, the market-data effect is much larger:
	the SLV fair insurance fee exceeds the LV fair insurance fee by approximately $28.06\%$, $26.15\%$
	and $23.84\%$, compared with only $1.43\%$, $1.58\%$ and $1.71\%$ in the
	synthetic case. The omitted root-finding residuals are at most
	$7.9\cdot10^{-5}$ in contract-value units and are negligible on the scale of
	these comparisons.
	
	The level of the fair insurance fee also differs sharply between the two experiments.
	For example, at $\kappa_s=0$ the LV fair insurance fee is $554.22$ basis points in
	the synthetic case but only $27.87$ basis points in the market-data case. An
	important contributor to this difference is the dividend treatment in the
	account convention: the synthetic experiment has $q_t\equiv0$ and hence
	$D_q(0,T)=1$, whereas the market-data specification incorporates the
	nonzero market-implied dividend-yield term structure through the factor
	$D_q(0,T)$ in the maturity value of the fee-deducted account. The two
	experiments also differ in their interest-rate inputs, local-volatility surfaces
	and Heston parameters, so the difference in fair-fee levels should not
	be attributed to the dividend effect alone. In both
	models the market-data fair insurance fee increases with $\kappa_s$, consistently
	with Proposition~\ref{prop:fee-monotone}. Thus the fair-fee comparison
	provides an additional indication that the impact of stochastic volatility on
	the surrenderable guarantee is not specific to the synthetic local-volatility
	surface.
	
	\FloatBarrier
	
	We next examine the market-data surrender boundaries for the two representative cases $(c,\kappa_s)=(3\%,0)$ and $(3\%,3\%)$, using $F_0=G=100$. Figure~\ref{fig:market-surrender-regions} reports the LV boundary together with five representative slices of the SLV boundary surface, spanning the displayed range of $V$ on the Heston tree. The right panel of each row shows the corresponding SLV surrender region in the $(F,V)$-plane at $t=5$ over $0\leq V\leq0.30$. The numerical boundaries are plotted without smoothing. Localised changes near curve knots may reflect both the piecewise-constant dividend-yield inputs and the finite grid; they are therefore interpreted qualitatively rather than as exact discontinuities of the continuous-time boundary.
	
	At $t=5$, the displayed SLV policy varies materially across values of $V$. The general pattern is that a higher level of $V$ lowers the policy-account boundary and expands the value of continuation, while increasing $\kappa_s$ shifts the surrender region upward. The displayed $V$ levels are nodes of the chosen, VSTOXX-informed normalisation of $V_t$ and should be interpreted subject to Remark~\ref{rem:variance-scale}. The market-data boundaries therefore provide a policy-level complement to the price and fair-fee comparisons, without relying on individual boundary coordinates that may move under grid refinement.
	
	\begin{figure}[htbp]
		\centering
		\begin{subfigure}{\textwidth}
			\centering
			\includegraphics[width=\textwidth]{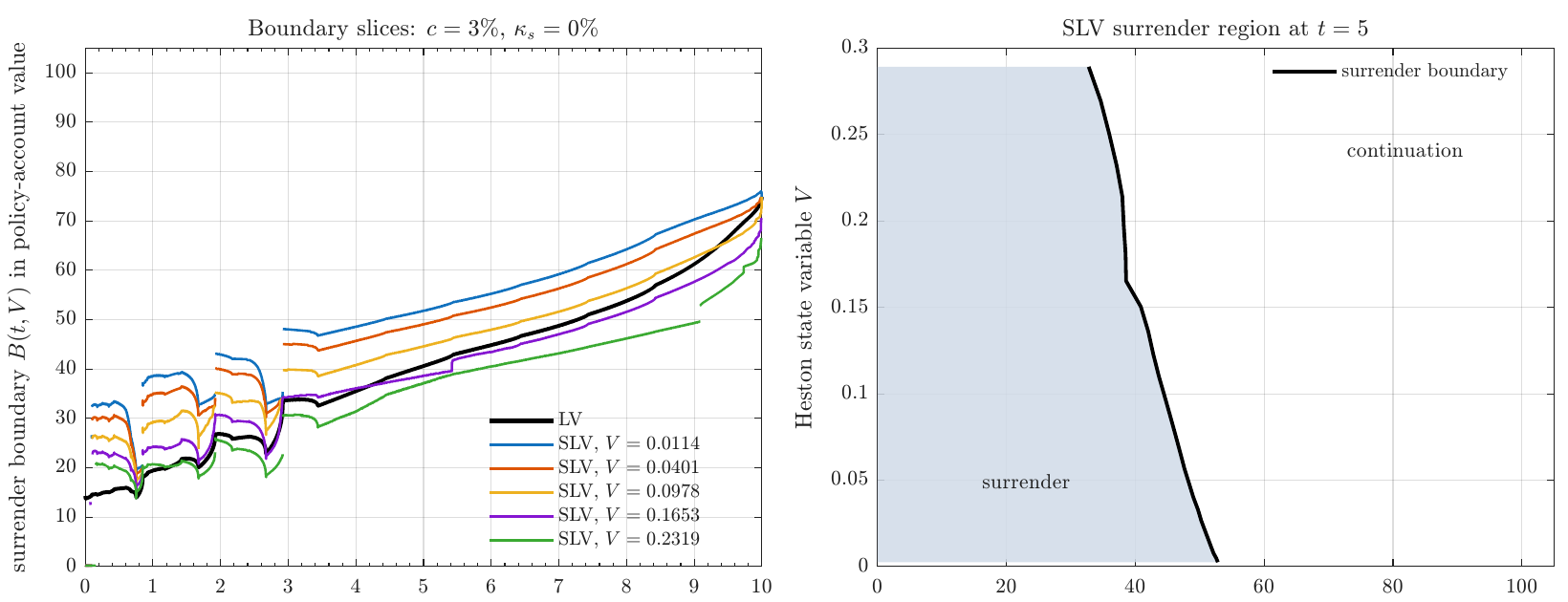}%
			\caption{$c=3\%$, $\kappa_s=0\%$.}
			\label{fig:market-surrender-kappa0}
		\end{subfigure}
		
		\vspace{0.3cm}
		
		\begin{subfigure}{\textwidth}
			\centering
			\includegraphics[width=\textwidth]{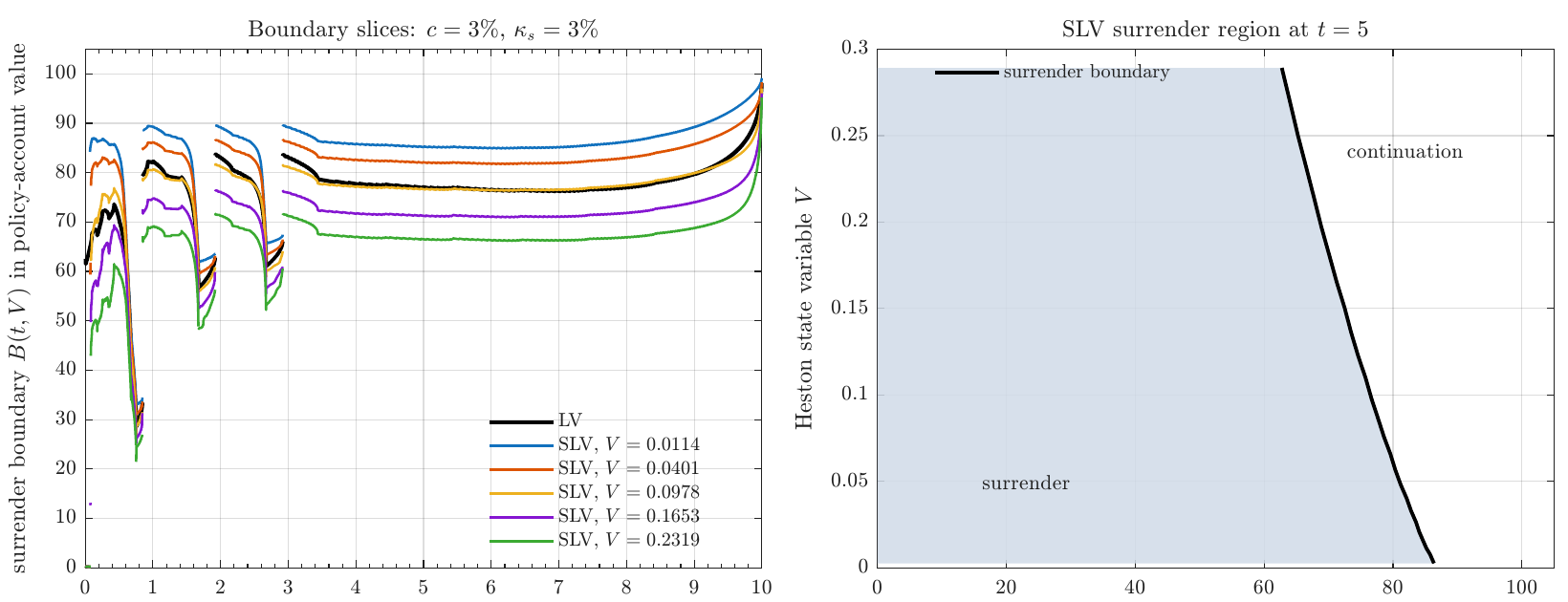}%
			\caption{$c=3\%$, $\kappa_s=3\%$.}
			\label{fig:market-surrender-kappa3}
		\end{subfigure}
		
		\caption{Optimal surrender boundaries and surrender regions under the EURO STOXX 50 market-informed specification. In each row, the left panel compares the LV boundary with five representative slices of the SLV boundary surface. The right panel reports the SLV surrender region in the $(F,V)$-plane at $t=5$. The shaded area is the surrender region and the white area is the continuation region.}
		\label{fig:market-surrender-regions}
	\end{figure}

	\FloatBarrier
	
	\section{Conclusion}
	\label{sec:conclusion}
	
	This paper studies GMMB guarantees with surrender options under local-volatility and Heston stochastic-local volatility models constrained to the same target local-volatility structure. The matched-target comparison is designed to isolate the model risk that remains after imposing the same one-dimensional marginal target at the model level, since continuation values and surrender decisions depend on conditional future dynamics that are not fully determined by the local-volatility projection.
	
	The hybrid tree/finite-difference approach provides a deterministic method for the Heston SLV setting with a calibrated leverage function. The numerical results show a clear distinction between terminal-only guarantees and guarantees with surrender: LV and SLV valuations are nearly indistinguishable in the former case, whereas economically relevant differences emerge once discrete-time optimal surrender is introduced.
	
	The market-data application to the EURO STOXX 50 provides a complementary market-informed test of the same mechanism. The GMMB comparison without surrender remains extremely tight, while the difference increases once the surrender option is introduced. Together with the volatility-dependent surrender regions, these results support the interpretation that continuation values and state-dependent surrender incentives are the main source of the observed model difference.
	
	\section*{Acknowledgements} 
	
	The authors report that no generative AI has been used for this paper, apart from suggestions for improvements in language and grammar (ChatGPT  5.6 Sol). The authors take full responsibility for the content of the publication.

	\appendix
	\setcounter{section}{0}
	\refstepcounter{section}
	
	\section*{Appendix \thesection. Numerical details of the Heston SLV hybrid method}
	\label{app:numerical-details}
	
	\phantomsection
	\addcontentsline{toc}{section}{Appendix \thesection. Numerical details of the Heston SLV hybrid method}
	
	\setcounter{equation}{0}
	\setcounter{figure}{0}
	\renewcommand{\theequation}{\thesection.\arabic{equation}}
	\renewcommand{\thefigure}{\thesection.\arabic{figure}}
	
	This appendix follows the order in which the Heston SLV algorithm is executed. Objects that are independent of the calibration iteration are constructed first: the time grid, the deterministic rate increments, the recombining Heston tree and the two spatial grids. The forward calibration then advances one time interval at a time. At each interval, the accepted leverage slice determines the transformed map and the coefficients of the one-dimensional forward equations; the density is propagated at every variance-tree node, redistributed through the tree, projected onto the stock grid and used to determine the next leverage slice. Once all leverage slices have been accepted, the surface is frozen and the same tree--PDE architecture is traversed backward to value the GMMB. The discussion below concerns the correlated branch $\rho\neq0$ unless stated otherwise. The independent branch uses the same Heston tree and the same leverage projection, but works directly on the log-stock grid and therefore omits the transformed map.
	\setcounter{subsection}{-1}
	\subsection{Time grid and deterministic inputs.}
	The reported hybrid calculations use the uniform time grid
	\begin{equation}
		t_i=i\Delta t,
		\qquad
		\Delta t=\frac{T}{N_t},
		\qquad
		i=0,\ldots,N_t.
		\label{eq:appendix-time-step}
	\end{equation}
	For each interval $[t_i,t_{i+1}]$, the deterministic curves enter through the exact interval averages
	\begin{equation}
		\bar r_i
		=
		\frac{1}{\Delta t}
		\int_{t_i}^{t_{i+1}} r_t\,\dd t,
		\qquad
		\bar q_i
		=
		\frac{1}{\Delta t}
		\int_{t_i}^{t_{i+1}} q_t\,\dd t,
		\label{eq:appendix-average-rates}
	\end{equation}
	so that
	\[
	D_r(t_i,t_{i+1})=e^{-\bar r_i\Delta t},
	\qquad
	D_q(t_i,t_{i+1})=e^{-\bar q_i\Delta t}.
	\]
	These interval quantities are computed once and are used in both the forward and backward sweeps.
	
	\subsection{Trinomial tree}
	\label{app:trinomial-tree}
	
	For the CIR process $V_t$, the exact first two conditional moments over one time step are
	\begin{align}
		M_1(v)
		&=
		\theta_V+e^{-\kappa_V\Delta t}(v-\theta_V),
		\nonumber\\
		M_2(v)
		&=
		M_1(v)^2
		+
		\frac{\omega^2(1-e^{-\kappa_V\Delta t})}{\kappa_V}
		\left[
		\frac12\theta_V(1-e^{-\kappa_V\Delta t})
		+ve^{-\kappa_V\Delta t}
		\right].
		\label{eq:appendix-cir-moments}
	\end{align}
	The implementation constructs a time-homogeneous square-root lattice with spacing
	\begin{equation}
		h_V=\frac34\omega\sqrt{\Delta t},
		\qquad
		v_n=
		\left[
		\max\left\{0,\sqrt{V_0}+nh_V\right\}
		\right]^2,
		\label{eq:appendix-variance-lattice}
	\end{equation}
	for a finite integer range of $n$ determined by the admissibility search. The node $n=0$ is exactly $V_0$, so the tree is rooted at the prescribed initial variance.
	
	At a current lattice node $v_n$, let $n_a$ be the first lattice index at or above the conditional mean $M_1(v_n)$ in square-root coordinates. The neighbouring indices used by the search are
	\begin{equation}
		n_b=n_a-1,
		\qquad
		n_c=n_a+1,
		\qquad
		n_d=n_a-2,
		\qquad
		n_e=n_a-3.
		\label{eq:appendix-tree-neighbours}
	\end{equation}
	For a candidate triplet of child indices $(k_0,k_1,k_2)$, let
	$\pi_{n\to k_c}$, $c\in\{0,1,2\}$, denote the corresponding transition
	probability from the parent node $n$ to the candidate child node $k_c$. They are
	
	\begin{equation}
		\sum_{c=0}^2\pi_{n\to k_c}=1,
		\qquad
		\sum_{c=0}^2\pi_{n\to k_c}v_{k_c}=M_1(v_n),
		\qquad
		\sum_{c=0}^2\pi_{n\to k_c}v_{k_c}^2=M_2(v_n).
		\label{eq:appendix-tree-moment-matching}
	\end{equation}
	For each current lattice node, the implementation tests the local child-index triplets
	\begin{equation}
		(n_a,n_b,n_d),\quad
		(n_a,n_b,n_e),\quad
		(n_a,n_d,n_e),\quad
		(n_a,n_b,n_c),
		\label{eq:appendix-tree-triplets}
	\end{equation}
	in this order. The first triplet with valid child indices and transition probabilities no smaller than the numerical tolerance is retained. If none of these four candidates is admissible, a local mixed fallback constructed from neighbouring up/down configurations is used. Tiny negative stored probabilities are clipped to zero, the child indices are restricted to the finite lattice, and the probabilities are renormalised. The lower and upper bounds are enlarged until all stored transitions are compatible with the finite recombining state set.
	
	After the admissibility search has been completed, the finite lattice is relabelled as $\{v_j\}$. For each parent $v_j$, the accepted child indices form the set $\mathcal C(j)$ used in Section~\ref{subsec:variance-tree}, and $\pi_{j\to h}$ denotes the probability assigned to child $h\in\mathcal C(j)$. The tree is now complete: the same stored child maps and probabilities will redistribute probability mass in the forward sweep and form conditional expectations in the backward sweep.
	
	\subsection{Spatial grids}
	\label{app:spatial-grids}
	
	The leverage surface is stored on a fixed log-stock grid. Let $[x_{\min},x_{\max}]$ be its physical domain and define
	\begin{equation}
		\Delta x=\frac{x_{\max}-x_{\min}}{N_S},
		\qquad
		x_\ell=x_{\min}+\ell\Delta x,
		\qquad
		S_\ell=e^{x_\ell},
		\qquad
		\ell=0,\ldots,N_S.
		\label{eq:appendix-x-grid}
	\end{equation}
	Linear interpolation on this grid is denoted by $\mathcal I_x$. In the market-data implementation, the local-volatility input is tabulated in log-moneyness $\log(S/S_0)$ and is evaluated by bilinear interpolation in $(t,\log(S/S_0))$, together with boundary-cell extrapolation and the volatility floor used by the calibration driver. The same spatial-domain convention is retained at every level of a grid-refinement experiment.
	
	In the synthetic driver, the symmetric log-stock half-width is generated by
	\begin{equation}
		\ell_x
		=
		a_x\sqrt{T}\sqrt{\log(1/\varepsilon_{\rm fd})}
		+
		\left|
		\int_0^T
		\left(r_t-q_t-\frac12a_x^2\right)\dd t
		\right|,
		\qquad
		x_{\min,\max}=\log S_0\mp\ell_x,
		\label{eq:appendix-synthetic-x-domain}
	\end{equation}
	where $a_x$ is used only to set the numerical domain and is not a parameter of the local-volatility model.
	
	The transformed coordinate is advanced on a separate uniform grid. Since $g(0,S_0)=0$, its initial value is
	\begin{equation}
		y_{\rm init}
		=V_0-\frac{\omega}{\rho}g(0,S_0)
		=V_0.
		\label{eq:appendix-yinit}
	\end{equation}
	With transformed-domain tail tolerance $\varepsilon_{\rm fd}\in(0,1)$, set
	\begin{equation}
		\sigma_Y
		=
		\left|\frac{\omega}{\rho}\right|
		\sqrt{1-\rho^2}
		\sqrt{\max\{V_0,\theta_V\}},
		\label{eq:appendix-sigmaY}
	\end{equation}
	and
	\begin{equation}
		\ell_Y
		=
		\sigma_Y\sqrt{T}\sqrt{2\log(1/\varepsilon_{\rm fd})}
		+|\kappa_V(\theta_V-V_0)|T
		+0.25\,\omega\sqrt{\max\{V_0,\theta_V\}}\,T.
		\label{eq:appendix-y-width}
	\end{equation}
	
	The transformed domain is
	\begin{equation}
		y_{\min}=y_{\rm init}-\ell_Y,
		\qquad
		y_{\max}=y_{\rm init}+\ell_Y,
		\label{eq:appendix-y-domain}
	\end{equation}
	and the corresponding grid is
	\begin{equation}
		\Delta y=\frac{y_{\max}-y_{\min}}{N_Y},
		\qquad
		y_m=y_{\min}+m\Delta y,
		\qquad
		m=0,\ldots,N_Y.
		\label{eq:appendix-y-grid}
	\end{equation}
	Linear interpolation on this grid is denoted by $\mathcal I_Y$. Let $\varepsilon_V>0$ be the positive floor used in the conditional-variance update, and let $\varepsilon_{\rm dens}>0$ denote the threshold below which the projected density denominator is treated as numerically negligible. The initial leverage slice is defined directly on the stock grid by
	\begin{equation}
		L^0_\ell=L(0,S_\ell)
		=\frac{\sigLV(0,S_\ell)}{\sqrt{\max\{V_0,\varepsilon_V\}}},
		\qquad \ell=0,\ldots,N_S.
		\label{eq:appendix-initial-leverage}
	\end{equation}
	We write $P^i_{j,m}$ for the forward density-like array at time $t_i$, variance node $v_j$ and transformed node $y_m$. The initial mass is placed at the exact Heston-tree root $V_0$ and at the interior $Y$-grid node nearest to $y_{\rm init}$. Equation~\eqref{eq:appendix-initial-leverage} and this point-mass initialisation complete the forward initial state.
	
	\subsection{Forward calibration over one time interval}
	\label{app:forward-calibration}
	
	Assume that the accepted density $P^i$ and leverage slice $L^i_\ell=L(t_i,S_\ell)$ are available at time $t_i$. One forward step must produce both the propagated density $P^{i+1}$ and the next accepted leverage slice $L^{i+1}$. The operations are performed in the following order.
	
	\subsubsection*{Step 1: construction of the leverage-dependent arrays}
	
	The transformed map is built from the accepted leverage slice. First compute the unshifted primitive
	\begin{equation}
		\Gamma^i_0=0,
		\qquad
		\Gamma^i_\ell=\Gamma^i_{\ell-1}
		+\frac{\Delta x}{2}
		\left(\frac{1}{L^i_{\ell-1}}+\frac{1}{L^i_\ell}\right),
		\qquad \ell=1,\ldots,N_S.
		\label{eq:appendix-G}
	\end{equation}
	With $x_*=\log S_0$, impose the anchoring convention $g(t_i,S_0)=0$ by setting
	\begin{equation}
		g^i_\ell
		=\Gamma^i_\ell-\mathcal I_x[\Gamma^i](x_*).
		\label{eq:appendix-g}
	\end{equation}
	The time derivative of the map satisfies
	\begin{equation}
		g_t(t,S)
		=-\int_{\log S_0}^{\log S}
		\frac{L_t(t,e^x)}{L^2(t,e^x)}\,\dd x.
		\label{eq:appendix-gt}
	\end{equation}
	For $i\geq1$, $L_t(t_i,S_\ell)$ is approximated by $(L^i_\ell-L^{i-1}_\ell)/\Delta t$ and the integral is evaluated by the trapezoidal rule; at $t_0$, the array $g_t$ is set to zero. Finally, $L_S$ is computed by centred differences at the interior stock-grid nodes and one-sided differences at the endpoints. The four arrays
	\[
	L^i,\qquad g^i,\qquad g_t^i,\qquad
	\Lambda^i:=L^i+S L_S^i
	\]
	are therefore available on the stock grid. They are the inputs to the map inversion and coefficient evaluation at the current time level.
	
	\subsubsection*{Step 2: inversion of the transformed map}
	
	For each pair $(y_m,v_j)$, the transformed relation gives
	\[
		y_m=v_j-\frac{\omega}{\rho}g(t_i,S).
	\]
	Thus the stock value attached to the tree--PDE grid point is found by solving
	\begin{equation}
		g(t_i,S)=\frac{\rho}{\omega}(v_j-y_m).
		\label{eq:appendix-inversion-target}
	\end{equation}
	The tabulated function $g(t_i,\cdot)$ is monotone on the log-stock grid. The target is bracketed between neighbouring entries of $g^i$, and linear interpolation in the corresponding log-stock coordinate gives
	\[
		S^i_{j,m}:=S(t_i,y_m,v_j).
	\]
	Targets outside the tabulated range are clamped to the boundary stock nodes. Once $S^i_{j,m}$ has been recovered, the same stock-grid interpolation supplies the local values of $L^i$, $g_t^i$ and $\Lambda^i=L^i+S L_S^i$. Hence all leverage-dependent quantities needed by the conditional PDE are now known at every $(y_m,v_j)$.
	
	\subsubsection*{Step 3: evaluation of the conditional forward coefficients}
	
	For the interval $[t_i,t_{i+1}]$, let $\mu_{Y,i}(S,V)$ denote the drift in 
	\eqref{eq:muY-main}, with $r_t-q_t$ replaced by the exact interval average $\bar r_i-\bar q_i$ and all leverage-dependent quantities frozen at the accepted time level $t_i$. Define
	\begin{equation}
		\mu^{i,j}_{Y,m}
		=\mu_{Y,i}(S^i_{j,m},v_j),
		\qquad
		\nu_j
		=\frac12\frac{\omega^2(1-\rho^2)}{\rho^2}v_j.
		\label{eq:appendix-forward-grid-coefficients}
	\end{equation}
	At a fixed tree node $v_j$, the diffusion coefficient $\nu_j$ is constant over the entire $Y$ grid, whereas the drift varies with $m$ through the recovered stock value $S^i_{j,m}$. The coefficient arrays are now complete and the numerical forward solve can begin.
	
	\subsubsection*{Step 4: one-dimensional finite-volume propagation, including substepping}
	
	At node $v_j$, the conditional density solves
	\begin{equation}
		\partial_t p_j
		=-\partial_y(\mu_jp_j)+\nu_j\partial_{yy}p_j,
		\qquad
		\mu_j(t_i,y_m)=\mu^{i,j}_{Y,m}.
		\label{eq:appendix-forward-pde}
	\end{equation}
	The number of temporal substeps is selected at this stage, after the coefficients have been evaluated. Define the local drift indicator
	\begin{equation}
		C_{i,j}
		=\max_{1\leq m\leq N_Y-1}
		\frac{|\mu^{i,j}_{Y,m}|\Delta t}{\Delta y}.
		\label{eq:appendix-courant}
	\end{equation}
	With target value $C_*=0.75$ and cap $n_{\max}$, set
	\begin{equation}
		n_s
		=\min\left\{n_{\max},
		\max\left(1,\left\lceil\frac{C_{i,j}}{C_*}\right\rceil\right)\right\},
		\qquad
		\Delta t_s=\frac{\Delta t}{n_s}.
		\label{eq:appendix-substeps}
	\end{equation}
	The reported forward calibration uses $n_{\max}=8$. The drift and diffusion coefficients are frozen during these substeps. To make the propagation explicit, write
	\[
		P^{i,0}_{j,m}=P^i_{j,m},
		\qquad
		P^{i,n_s}_{j,m}=\widetilde P^{i+1}_{j,m},
	\]
	where $\widetilde P^{i+1}_{j,m}$ denotes the density after the conditional $Y$ solve and before the tree transition.
	
	For one substep, define
	\begin{equation}
		\mu_{m+1/2}
		=\frac12\left(\mu^{i,j}_{Y,m+1}+\mu^{i,j}_{Y,m}\right),
		\qquad
		\mu^+=\max(\mu,0),
		\qquad
		\mu^-=\min(\mu,0).
	\end{equation}
	The conservative upwind--diffusion flux is
	\begin{equation}
		\mathcal F_{m+1/2}
		=\mu^+_{m+1/2}p_m
		+\mu^-_{m+1/2}p_{m+1}
		-\nu_j\frac{p_{m+1}-p_m}{\Delta y}.
		\label{eq:appendix-forward-flux}
	\end{equation}
	The fully implicit conservative balance over one substep is
	\begin{equation}
		P^{i,s+1}_{j,m}
		+\lambda_s\left[
		\mathcal F_{j,m+1/2}\!\left(P^{i,s+1}_{j,\cdot}\right)
		-\mathcal F_{j,m-1/2}\!\left(P^{i,s+1}_{j,\cdot}\right)
		\right]
		=P^{i,s}_{j,m},
		\qquad
		\lambda_s=\frac{\Delta t_s}{\Delta y}.
		\label{eq:appendix-forward-balance}
	\end{equation}
	Let
	\[
		\delta_j=\frac{\nu_j\Delta t_s}{\Delta y^2}.
	\]
	Substituting the flux into \eqref{eq:appendix-forward-balance} gives the fully implicit interior coefficients
	\begin{align}
		a^F_{j,m}&=-\delta_j-\lambda_s\mu^+_{m-1/2},\nonumber\\
		b^F_{j,m}&=1+2\delta_j
		+\lambda_s\left(\mu^+_{m+1/2}-\mu^-_{m-1/2}\right),\nonumber\\
		c^F_{j,m}&=-\delta_j+\lambda_s\mu^-_{m+1/2}.
		\label{eq:appendix-forward-coefficients}
	\end{align}
	Thus each substep maps $P^{i,s}_{j,\cdot}$ to $P^{i,s+1}_{j,\cdot}$ through
	\[
		a^F_{j,m}P^{i,s+1}_{j,m-1}
		+b^F_{j,m}P^{i,s+1}_{j,m}
		+c^F_{j,m}P^{i,s+1}_{j,m+1}
		=P^{i,s}_{j,m},
		\qquad s=0,\ldots,n_s-1.
	\]
	Zero total flux is imposed at $y_{\min}$ and $y_{\max}$. After the final substep, the vector $\widetilde P^{i+1}_{j,\cdot}$ is ready to be passed from the current variance node to its three children.
	
	\subsubsection*{Step 5: redistribution through the Heston tree}
	
	For each child node $v_h$ at time $t_{i+1}$, the new joint density is obtained from the redistribution formula \eqref{eq:forward-tree-redistribution-main}. This completes the propagation of the joint density from $t_i$ to $t_{i+1}$. The resulting array $P^{i+1}$ is the input to the stock-grid projection. It is not propagated again during the projection iterations used to determine $L^{i+1}$.
	
	\subsubsection*{Step 6: projection onto the stock grid and conditional moment}
	
	For each new-time tree node, define the transformed-grid interpolant
	\begin{equation}
		\widehat P^{i+1}_j(y)
		=\mathcal I_Y[P^{i+1}_{j,\cdot}](y).
		\label{eq:appendix-P-interpolant}
	\end{equation}
	At projection iteration $n$, let $g^{i+1,[n]}$ be the current guess for the new-time transformed map. At stock node $S_\ell$ and variance node $v_j$, this map determines
	\begin{equation}
		y^{i+1,[n]}_{j,\ell}
		=v_j-\frac{\omega}{\rho}g^{i+1,[n]}(S_\ell).
		\label{eq:appendix-y-stock}
	\end{equation}
	Because $y^{i+1,[n]}_{j,\ell}$ does not generally coincide with a $Y$-grid node, the density is evaluated there by interpolation. Only coordinates strictly inside $[y_{\min},y_{\max}]$ contribute, and non-positive interpolated values are omitted when the conditional moment is formed. The estimate of $\E[V_{t_{i+1}}\mid S_{t_{i+1}}=S_\ell]$ is the quantity $\widehat v^{i+1,[n]}(S_\ell)$ defined in \eqref{eq:conditional-variance-main}. At fixed $S_\ell$, the Jacobian of the transformation is independent of $v_j$ and cancels from the ratio. If the denominator in \eqref{eq:conditional-variance-main} is no larger than $\varepsilon_{\rm dens}$, the conditional moment is replaced by the unconditional mean of $V_{t_{i+1}}$ represented by the positive part of the current joint density.
	
	\subsubsection*{Step 7: leverage update, projection iteration and acceptance}
	
	The direct Markovian-projection update associated with iteration $n$ is given by \eqref{eq:discrete-leverage-update}.
	The dependence is implicit because the projection uses the new-time map, while that map is itself constructed from the new leverage slice:
	\[
		L^{i+1,[n]}
		\longrightarrow g^{i+1,[n]}
		\longrightarrow \widehat v^{i+1,[n]}
		\longrightarrow \widehat L^{i+1,[n+1]}.
	\]
	The implementation resolves this dependence without repeating the forward transport step. After $P^{i+1}$ has been propagated and redistributed once, set $g^{i+1,[0]}:=g^i$. The first projection produces the direct update
	$L^{i+1,[1]}:=\widehat L^{i+1,[1]}$, from which the arrays
	$g^{i+1,[1]}$, $g_t^{i+1,[1]}$ and
	$L^{i+1,[1]}+S L_S^{i+1,[1]}$ are rebuilt. For every subsequent pass,
	$P^{i+1}$ remains fixed and
	\begin{equation}
		L^{i+1,[n+1]}
		=
		\alpha\widehat L^{i+1,[n+1]}
		+(1-\alpha)L^{i+1,[n]},
		\qquad n\geq1.
		\label{eq:appendix-relaxation}
	\end{equation}
	The updated map is then used in the next evaluation of
	\eqref{eq:conditional-variance-main}. After the final pass, the leverage slice and its associated auxiliary arrays are accepted, $P^{i+1}$ becomes the current density, and the algorithm advances to $[t_{i+1},t_{i+2}]$. Repeating Steps 1--7 over all intervals produces the calibrated leverage surface on the full time--stock grid.
	
	\subsection{Independent branch}
	\label{app:independent-branch}
	
	When $\rho=0$, the order of the forward algorithm is unchanged, but the transformed-map operations disappear. The spatial coordinate is $x=\log S$, the factorwise coefficients are those in 
	\eqref{eq:rho0-coefficients-main}, and the conditional equation 
	\eqref{eq:forward-rho0-main} is advanced directly on the common log-stock grid. There is therefore no need to construct $g$ and $g_t$, invert $Y\mapsto S$, or interpolate the propagated density back to the stock grid.
	
	Unlike the correlated equation in $Y$, the diffusion coefficient
	\[
		\nu^{i,j}_m
		=\frac12L^2(t_i,e^{x_m})v_j
	\]
	varies across the log-stock grid. The independent-branch forward solver therefore uses the conservative cell-face flux
	\begin{equation}
		\mathcal F^x_{m+1/2}
		=(\mu^x_{m+1/2})^+p_m
		+(\mu^x_{m+1/2})^-p_{m+1}
		-\frac{\nu^{i,j}_{m+1}p_{m+1}-\nu^{i,j}_{m}p_m}{\Delta x},
		\label{eq:appendix-rho0-forward-flux}
	\end{equation}
	where $\mu^x_{m+1/2}$ is the arithmetic average of the neighbouring drift values. Thus the diffusion term is discretised through neighbouring values of $\nu p$ and is not treated as spatially constant. The same implicit conservative balance as in
	\eqref{eq:appendix-forward-balance} is then applied with $\Delta y$ replaced by $\Delta x$.
	
	After the log-stock solve, the density is redistributed through the same Heston tree, and the conditional moment is read directly from the common grid as
	\begin{equation}
		\widehat v^{i+1}_m
		=
		\frac{\sum_jv_jP^{i+1}_{j,m}}
		{\sum_jP^{i+1}_{j,m}}.
	\end{equation}
	The leverage update and the advance to the next time level then proceed as in Step 7. Thus the independent branch has the same forward-calibration logic, but not the map inversion and stock-grid projection required by the correlated branch. Its backward discretisation is given explicitly in Subsection~\ref{app:backward-pricing}.
	
	\subsection{Backward pricing on the calibrated leverage surface}
	\label{app:backward-pricing}
	
	After the forward loop has reached $T$, every leverage slice and its auxiliary arrays are fixed. The backward pass now reverses the splitting: at each parent variance node it first averages the already known child values and then solves the one-dimensional continuation equation over the preceding time interval.
	
	Define the stock value represented by a spatial grid point as
	\begin{equation}
		\mathsf S(t_i,\zeta_m,v_j)
		=
		\begin{cases}
			e^{x_m}, & \rho=0,\\
			S(t_i,y_m,v_j), & \rho\neq0.
		\end{cases}
		\label{eq:appendix-recovered-stock}
	\end{equation}
	The backward array is initialised at maturity by
	\begin{equation}
		U^{N_t}_{j,m}
		=
		\left(G-F_0\frac{\mathsf S(T,\zeta_m,v_j)}{S_0}e^{-cT}\right)^+.
		\label{eq:appendix-terminal-array}
	\end{equation}
	
	For $i=N_t-1,\ldots,0$, the following sequence is repeated.
	
	\subsubsection*{Step B1: conditional expectation over the tree}
	
	At parent node $v_j$, form the conditional expectation $R^i_{j,m}$ in \eqref{eq:tree-expectation-main}. This vector is the terminal datum for the spatial continuation solve over $[t_i,t_{i+1}]$.
	
	\subsubsection*{Step B2: coefficient evaluation and backward substeps}
	
	In the correlated branch, the stored leverage slice at $t_i$ is used to recover $S(t_i,y_m,v_j)$ and to evaluate
	\begin{equation}
		\partial_tu_j
		+\mu_j\partial_yu_j
		+\nu_j\partial_{yy}u_j
		-r_tu_j=0,
		\qquad
		\mu_j(t_i,y_m)=\mu_Y(t_i,S(t_i,y_m,v_j),v_j).
		\label{eq:appendix-backward-pde}
	\end{equation}
	The same indicator 
	\eqref{eq:appendix-courant} and target $C_*=0.75$ select equal implicit substeps, now with cap $n_{\max}=4$. For $\Delta t_s=\Delta t/n_s$, the undiscounted centred interior coefficients are
	\begin{align}
		a_{j,m}
		&=-\nu_j\frac{\Delta t_s}{\Delta y^2}
		+\mu^{i,j}_{Y,m}\frac{\Delta t_s}{2\Delta y},\nonumber\\
		b_{j,m}
		&=1+2\nu_j\frac{\Delta t_s}{\Delta y^2},\nonumber\\
		c_{j,m}
		&=-\nu_j\frac{\Delta t_s}{\Delta y^2}
		-\mu^{i,j}_{Y,m}\frac{\Delta t_s}{2\Delta y}.
		\label{eq:appendix-backward-coefficients}
	\end{align}
	For $\rho=0$, define at node $v_j$
	\[
		\mu^{i,j}_{x,m}
		=\bar r_i-\bar q_i-\frac12L^2(t_i,e^{x_m})v_j,
		\qquad
		\nu^{i,j}_{x,m}
		=\frac12L^2(t_i,e^{x_m})v_j,
	\]
	where $\bar r_i$ and $\bar q_i$ are the interval averages defined above. The local drift indicator is computed with $\Delta x$, and the centred implicit coefficients for the backward equation are
	\begin{align}
		a^x_{j,m}
		&=-\nu^{i,j}_{x,m}\frac{\Delta t_s}{\Delta x^2}
		+\mu^{i,j}_{x,m}\frac{\Delta t_s}{2\Delta x},\nonumber\\
		b^x_{j,m}
		&=1+2\nu^{i,j}_{x,m}\frac{\Delta t_s}{\Delta x^2},\nonumber\\
		c^x_{j,m}
		&=-\nu^{i,j}_{x,m}\frac{\Delta t_s}{\Delta x^2}
		-\mu^{i,j}_{x,m}\frac{\Delta t_s}{2\Delta x}.
		\label{eq:appendix-backward-rho0-coefficients}
	\end{align}
	At this point the coefficient matrix and the number of substeps are fixed. Before the systems are solved, the financial boundary values must be supplied.
	
	\subsubsection*{Step B3: boundary conditions}
	
	Boundary values are imposed in the recovered stock variable. Let $\zeta_{\min}=x_{\min}$ and $\zeta_{\max}=x_{\max}$ for $\rho=0$, and $\zeta_{\min}=y_{\min}$ and $\zeta_{\max}=y_{\max}$ for $\rho\neq0$. Define
	\begin{equation}
		S_L^{i,j}=\mathsf S(t_i,\zeta_{\min},v_j),
		\qquad
		S_R^{i,j}=\mathsf S(t_i,\zeta_{\max},v_j),
		\qquad
		A_T=\frac{F_0}{S_0}e^{-cT}.
		\label{eq:appendix-stock-boundaries}
	\end{equation}
	For the terminal-only put-like guarantee, the asymptotic values are
	\begin{equation}
		B^{i,j}_{\xi}
		=
		\left(GD_r(t_i,T)-A_TS^{i,j}_{\xi}D_q(t_i,T)\right)^+,
		\qquad \xi\in\{L,R\}.
		\label{eq:appendix-boundary}
	\end{equation}
	For the surrenderable guarantee, the right boundary is zero. At an admissible surrender date $t_i>0$, the left boundary is the obstacle,
	\begin{equation}
		B^{i,j}_L
		=
		\left(G-e^{-\kappa_s(T-t_i)}F_0
		\frac{S_L^{i,j}}{S_0}e^{-ct_i}\right)^+.
		\label{eq:appendix-surrender-left-boundary}
	\end{equation}
	At a date at which surrender is not allowed, the low-stock linear asymptote is propagated from the next grid date $u=t_{i+1}$. With
	\begin{equation}
		\beta_u
		=\frac{F_0}{S_0}e^{-cu}e^{-\kappa_s(T-u)},
		\label{eq:appendix-surrender-beta}
	\end{equation}
	the boundary values are
	\begin{equation}
		B^{i,j}_L
		=
		\left(GD_r(t_i,u)-\beta_uS_L^{i,j}D_q(t_i,u)\right)^+,
		\qquad
		B^{i,j}_R=0.
		\label{eq:appendix-continuation-left-boundary}
	\end{equation}
	Under the reported discrete-time surrender convention, every positive pre-maturity grid date is admissible and $t_0=0$ is the only such continuation date. Since discounting is applied after the spatial solve, prescribed time-$t_i$ boundary values are divided by $D_r(t_i,t_{i+1})$ when they are inserted into the undiscounted tridiagonal system.
	
	\subsubsection*{Step B4: implicit continuation solve and discounting}
	
	The right-hand side is $R^i_{j,m}$ at the first substep and the preceding substep solution thereafter. At each substep, the known boundary contributions from Step B3 are moved to the right-hand side and the tridiagonal system is solved by the Thomas algorithm. After the last spatial substep, the exact factor $D_r(t_i,t_{i+1})$ is applied to obtain the continuation value $\widetilde U^i_{j,m}$. This substepping improves temporal resolution of the transformed transport term; it is not, by itself, a proof of an M-matrix property for the centred backward convection discretisation.
	
	\subsubsection*{Step B5: obstacle projection and backward advance}
	
	At each admissible surrender date, the continuation value is projected onto the immediate surrender payoff,
	\begin{equation}
		U^i_{j,m}
		=
		\max\left\{
		\widetilde U^i_{j,m},
		\left(G-e^{-\kappa_s(T-t_i)}F_0
		\frac{\mathsf S(t_i,\zeta_m,v_j)}{S_0}e^{-ct_i}\right)^+
		\right\}.
		\label{eq:appendix-obstacle}
	\end{equation}
	The same projection is imposed on the boundary values at surrender dates, whereas no obstacle projection is applied at $t_0=0$. The resulting array $U^i$ becomes the input to Step B1 at the preceding time level. Repetition down to $t_0$ completes the backward valuation. For each variance node, the surrender boundary is extracted from the transition between equality with the obstacle and strict continuation, yielding the state-dependent surface $B^{\mathrm{SLV}}(t,V)$ reported in Section~\ref{sec:numerics}.

	\subsection{Local-volatility benchmark}
	\label{app:lv-benchmark}
	
	The LV benchmark is solved only backward in time on the log-stock grid introduced in Subsection~\ref{app:spatial-grids}. Writing $x=\log S$, define
	\[
		\mu_{\rm LV}(t,x)
		=r_t-q_t-\frac12\sigLV^2(t,e^x),
		\qquad
		\nu_{\rm LV}(t,x)
		=\frac12\sigLV^2(t,e^x).
	\]
	Between admissible surrender dates, the value satisfies
	\begin{equation}
		\partial_tu
		+\mu_{\rm LV}(t,x)\partial_xu
		+\nu_{\rm LV}(t,x)\partial_{xx}u
		-r_tu=0.
		\label{eq:appendix-lv-backward}
	\end{equation}
	The reported implementation uses the same time grid and a one-dimensional tridiagonal implicit backward discretisation. The terminal payoff is the one-factor version of \eqref{eq:appendix-terminal-array}; exact interval discounting is applied after the spatial solve. The terminal-only boundary values are obtained from \eqref{eq:appendix-boundary} with the stock endpoints $e^{x_{\min}}$ and $e^{x_{\max}}$. For the surrenderable guarantee, the high-stock boundary is zero, the low-stock boundary is the immediate obstacle at admissible dates, and the continuation-date asymptote is the one-factor counterpart of \eqref{eq:appendix-continuation-left-boundary}. The obstacle projection is imposed at every positive admissible surrender date and omitted at $t_0=0$. No forward density propagation is required for this benchmark.

	\refstepcounter{section}
	
	\section*{Appendix \thesection. Proof of Proposition~\ref{prop:fee-monotone}}
	\label{app:proof-fee-monotonicity}
	
	\phantomsection
	\addcontentsline{toc}{section}{Appendix \thesection. Proof of Proposition~\ref{prop:fee-monotone}}
	
	\setcounter{equation}{0}
	\renewcommand{\theequation}{\thesection.\arabic{equation}}
	
	\begin{proof}
		All the expectations are finite. Indeed,
		\[
		0\leq\Phi_{\kappa_s}(t)\leq G,
		\]
		and therefore
		\[
		0\leq U_{\mathrm{Sur}}(0;c,\kappa_s)\leq \overline{D}\,G.
		\]
		No separate terminal payoff needs to be specified, since
		\[
		\Phi_{\kappa_s}(T)=\left(G-F_T^{(c)}\right)^+,
		\]
		which is independent of $\kappa_s$.
		
		\emph{(i)} Fix $\kappa_2>\kappa_1$. For every $t\leq T$,
		\[
		e^{-\kappa_2(T-t)}F_t^{(c)}
		\leq
		e^{-\kappa_1(T-t)}F_t^{(c)},
		\]
		because $T-t\geq0$ and $F_t^{(c)}>0$. Since
		$x\mapsto(G-x)^+$ is nonincreasing, it follows that
		\begin{equation}
			\label{eq:pathwise}
			\Phi_{\kappa_2}(t)\geq\Phi_{\kappa_1}(t)
			\qquad\text{pathwise for every }t\leq T.
		\end{equation}
		For $\tau\in\mathcal T^{\mathcal D}_{0,T}$, define
		\[
		J_\tau(\kappa_s)
		=
		\E^{\Q}\!\left[D_r(0,\tau)\Phi_{\kappa_s}(\tau)\right].
		\]
		Evaluating \eqref{eq:pathwise} at $t=\tau$, multiplying by the positive
		discount factor and taking expectations gives
		$J_\tau(\kappa_2)\geq J_\tau(\kappa_1)$. Since the two suprema are taken
		over the same set of admissible stopping times,
		\[
		U_{\mathrm{Sur}}(0;c,\kappa_2)
		\geq
		U_{\mathrm{Sur}}(0;c,\kappa_1).
		\]
		
		\emph{(ii)} On the event $\mathcal E$, we have $\tau^*<T$, and hence
		\[
		e^{-\kappa_1(T-\tau^*)}>e^{-\kappa_2(T-\tau^*)}.
		\]
		Moreover, the definition of $\mathcal E$ implies that both payoffs lie in
		the region where the positive part is active. Therefore, on $\mathcal E$,
		\begin{align*}
		\Phi_{\kappa_2}(\tau^*)-\Phi_{\kappa_1}(\tau^*)
		&=F_{\tau^*}^{(c)}
		\left[e^{-\kappa_1(T-\tau^*)}-e^{-\kappa_2(T-\tau^*)}\right]\\
		&>0.
		\end{align*}
		Outside $\mathcal E$, the difference is nonnegative by
		\eqref{eq:pathwise}. Since $D_r(0,\tau^*)>0$ and $\Q(\mathcal E)>0$,
		\[
		J_{\tau^*}(\kappa_2)>J_{\tau^*}(\kappa_1).
		\]
		Using the optimality of $\tau^*$ for $\kappa_1$, we obtain
		\[
		U_{\mathrm{Sur}}(0;c,\kappa_2)
		\geq J_{\tau^*}(\kappa_2)
		>J_{\tau^*}(\kappa_1)
		=U_{\mathrm{Sur}}(0;c,\kappa_1).
		\]
		
		\emph{(iii)} Let $\kappa_2>\kappa_1$ and write
		\[
		c_1=c^*(\kappa_1),
		\qquad
		c_2=c^*(\kappa_2).
		\]
		Since $\Gamma(c)$ does not depend on $\kappa_s$, part~(i) gives
		\[
		H(c,\kappa_2)\geq H(c,\kappa_1)
		\]
		for every fixed $c$. In particular,
		\[
		H(c_1,\kappa_2)\geq H(c_1,\kappa_1)=0.
		\]
		If $c_2<c_1$, the strict decrease of $c\mapsto H(c,\kappa_2)$ would imply
		\[
		H(c_1,\kappa_2)<H(c_2,\kappa_2)=0,
		\]
		which is a contradiction. Hence $c_2\geq c_1$.
		
		Under the additional hypotheses of the strict statement in part~(iii),
		part~(ii), evaluated at $c=c_1$, yields
		\[
		H(c_1,\kappa_2)>H(c_1,\kappa_1)=0.
		\]
		If $c_2\leq c_1$, the decrease of $H(\cdot,\kappa_2)$ would instead give
		\[
		H(c_1,\kappa_2)\leq H(c_2,\kappa_2)=0,
		\]
		again a contradiction. Therefore $c_2>c_1$.
	\end{proof}

	

\end{document}